\documentclass[11pt,reqno]{amsart}
\usepackage[left=2.15cm,right=2.15cm,top=2.05cm,bottom=2.25cm]{geometry}
\usepackage{amsmath,amssymb,amsthm,mathtools,bm,mathrsfs,dsfont}
\usepackage[mathscr]{eucal}
\usepackage{amsaddr}
\usepackage{microtype}
\usepackage{enumitem}
\usepackage{booktabs}
\usepackage[hidelinks]{hyperref}
\usepackage[nameinlink,capitalize,noabbrev]{cleveref}
\usepackage{aliascnt}
\usepackage{subcaption}
\usepackage{tikz}
\usetikzlibrary{decorations.pathreplacing, calligraphy}
\usepackage{etoolbox}

\makeatletter
\patchcmd{\@tocline}
{\hfil}
{\dotfill}
{}{}
\makeatother

\allowdisplaybreaks
\setlist[itemize]{leftmargin=2em,itemsep=0.25em,topsep=0.35em}
\setlist[enumerate]{leftmargin=2.4em,itemsep=0.35em,topsep=0.35em}

\newtheorem{theorem}{Theorem}[section]
\newaliascnt{proposition}{theorem}
\newtheorem{proposition}[proposition]{Proposition}
\aliascntresetthe{proposition}
\newaliascnt{lemma}{theorem}
\newtheorem{lemma}[lemma]{Lemma}
\aliascntresetthe{lemma}
\newaliascnt{corollary}{theorem}
\newtheorem{corollary}[corollary]{Corollary}
\aliascntresetthe{corollary}
\theoremstyle{definition}
\newaliascnt{remark}{theorem}
\newtheorem{remark}[remark]{Remark}
\aliascntresetthe{remark}

\crefname{theorem}{Theorem}{Theorems}
\crefname{proposition}{Proposition}{Propositions}
\crefname{lemma}{Lemma}{Lemmas}
\crefname{corollary}{Corollary}{Corollaries}
\crefname{remark}{Remark}{Remarks}

\DeclareMathOperator{\Tr}{Tr}
\DeclareMathOperator{\supp}{supp}
\newcommand{\I}{\mathbf 1}
\newcommand{\A}{\mathsf{A}}
\newcommand{\B}{\mathsf{B}}
\newcommand{\Y}{\mathsf{Y}}

\newcommand{\states}[1]{\mathcal{S}(#1)}
\newcommand{\dd}{\mathrm d}
\newcommand{\Breg}{\mathfrak{B}}
\newcommand{\norm}[1]{\left\lVert#1\right\rVert}
\newcommand{\abs}[1]{\left\lvert#1\right\rvert}
\newcommand{\sand}{\widetilde}

\hypersetup{
  pdftitle={Sharp Continuity of Petz and Sandwiched R\'enyi Conditional Entropies},
  pdfauthor={Hao-Chung Cheng and Po-Chieh Liu},
  pdfsubject={Quantum information theory},
  pdfkeywords={R\'enyi conditional entropy, continuity bound,
    Petz--R\'enyi divergence, sandwiched Renyi divergence}
}

\begin{document}

\let\origmaketitle\maketitle
\def\maketitle{
	\begingroup
	\def\uppercasenonmath##1{} 
	\let\MakeUppercase\relax 
	\origmaketitle
	\endgroup
}

\title[]{\bfseries \Large{ 
		Sharp Continuity of Petz and Sandwiched R\'enyi Conditional Entropies
}}

\author{ \normalsize 
	{Hao-Chung Cheng}$^{1\text{--}5}$
    and
	{Po-Chieh Liu}$^{1,2}$
}
\address{\small  	
	$^1$Department of Electrical Engineering and Graduate Institute of Communication Engineering,\\ National Taiwan University, Taipei 106, Taiwan (R.O.C.)\\
	$^2$Department of Mathematics, National Taiwan University\\
	$^3$Center for Quantum Science and Engineering, National Taiwan University\\
	$^4$Hon Hai (Foxconn) Quantum Computing Center, New Taipei City 236, Taiwan (R.O.C.)\\
	$^5$Physics/Mathematics Divisions, National Center for Theoretical Sciences, Taipei 10617, Taiwan (R.O.C.)
}

\email{\href{mailto:haochung.ch@gmail.com}{haochung.ch@gmail.com}}

\date{}


\begin{abstract}
We determine the sharp modulus of continuity, in trace distance, of the optimized Petz and sandwiched R\'enyi conditional entropies for every order $\alpha\in[\frac12,1)$.
If two bipartite states are within trace distance $\delta$,  then both conditional entropies differ by at most
$\frac{1}{1-\alpha}
 \log[(1-\varepsilon)^{\alpha}
 +(D-1)^{1-\alpha}\varepsilon^{\alpha}]$, where $\varepsilon := \min\{\delta,1-1/D\}$ and $D$ is the effective dimension, given by the dimension of the first
subsystem times the largest possible Schmidt rank.
For every distance constraint $\delta\in[0,1]$, the bound is attained by
an isotropic pair with a maximally entangled anchor.
Taking $\alpha\uparrow1$ recovers the recent sharp continuity bound of quantum conditional entropy by Berta \emph{et al.} [\href{https://arxiv.org/abs/2607.24687}{arXiv:2607.24687}].

The proof linearizes the relevant concave R\'enyi functional at a
comparison point dictated by the isotropic equality family.
Schmidt-rank domination extends the equality geometry to an arbitrary anchor state, after which trace-distance duality and a noncommutative calibration estimate control the perturbation and anchor term without weakening the sharp constant.
The latter estimate requires matrix analysis and is assisted by ChatGPT 5.6 Sol.
\end{abstract}

\maketitle

\vspace{-20pt}

\tableofcontents

\section{Introduction}\label{sec:introduction}

Continuity bounds quantify the stability of information measures under
perturbations of the underlying state.  The study began with Fannes'
inequality and the sharp entropy bound of Audenaert
\cite{Fannes1973,Audenaert2007}.  For quantum conditional entropy,
Alicki and Fannes obtained the first dimensionally useful estimate, which
was later substantially strengthened by Winter
\cite{AlickiFannes2004,Winter2016}.  Sharp bounds were subsequently
obtained in the classical setting and when the conditioning system is
classical \cite{AlhejjiSmith2020,Wilde2020}.  For fully quantum states,
the sharp continuity bound was first established under the fixed-marginal
assumption $\rho_{\B}=\sigma_{\B}$ in the independent works
\cite{BertaLamiTomamichel2025,AudenaertEtAl2025}; the unrestricted case,
together with an effective-dimension refinement, was resolved recently in
\cite{BertaEtAl2026}.

Conditional R\'enyi entropies are parametric generalizations of the conditional quantum entropy.
They are operationally relevant because they characterize error and
strong-converse exponents in several quantum-information tasks.
Petz--R\'enyi conditional entropies of orders $\alpha\in[\frac12,1)$ appear in direct
error-exponent analyses for classical data compression with quantum side
information \cite{CHDH-2018, Ren25, preparation} and, more recently, for quantum state merging,
entanglement distillation, and quantum communication
\cite{BertaChengYao2026}. 
Sandwiched R\'enyi quantities of orders
$\alpha\in[\frac12,1)$ characterize reliability and strong-converse
exponents in quantum decoupling, privacy amplification, and related
information-processing tasks
\cite{LY21a,LY24a,BertaChengYao2026,CDG24,RegulaTomamichel2026}.  
Sharp continuity bounds thus
provide robustness estimates for quantities that directly govern
the exponential behavior of these protocols.

For R\'enyi conditional entropy, Jabbour and Datta proved the sharp
continuity bound for classical distributions and for a classical
conditioning system, throughout $0\leq\alpha<1$
\cite{JabbourDatta2022}.  
For fully quantum states, continuity bounds for
the sandwiched R\'enyi conditional entropy for $\alpha \geq 1/2$ and related quantities were
studied in \cite{MarwahDupuis2022,BeigiGoodarzi2023,BluhmCapelGondolfMoebus2026}.
In this paper, we determine the sharp
fully quantum continuity bounds for both the optimized Petz and sandwiched
R\'enyi conditional entropies for every
$\alpha\in[\frac12,1)$.

Let $d_{\A}$ and $d_{\B}$ denote the dimensions of systems $\A$ and
$\B$, respectively, and let
$D:=d_{\A}\min\{d_{\A},d_{\B}\}$ be the effective dimension.
If two quantum states $\rho_{\A\B}$ and $\sigma_{\A\B}$ are $\delta$-close in trace distance, then the absolute differences of both the
optimized Petz and sandwiched R\'enyi conditional entropies are bounded by
\begin{equation}\label{eq:introduction-modulus}
 \begin{cases}
 \displaystyle
 \frac{1}{1-\alpha}
 \log\!\left[
 (1-\delta)^\alpha
 +(D-1)^{1-\alpha}\delta^\alpha
 \right],
 &
 0\leq\delta\leq1-\dfrac1D,
 \\[3mm]
 \log D,
 &
 1-\dfrac1D<\delta\leq1.
 \end{cases}
\end{equation}
The bound is optimal for every distance constraint $\delta \in [0,1]$.  

Let $\Phi_r$ be a maximally entangled state of Schmidt rank
$r=\min\{d_{\A},d_{\B}\}$, and let $
 P_{\B_0}:=r\Tr_{\A}[\Phi_r]$
be the projection onto the support of its $\B$-marginal.
For $0\leq\delta\leq1-D^{-1}$, equality is attained by
\begin{align}
 \rho_{\A\B}
 &=
 (1-\delta)\Phi_r
 +\frac{\delta}{D-1}
 \left(
 \I_{\A}\otimes P_{\B_0}-\Phi_r
 \right),
 \qquad
 \sigma_{\A\B}=\Phi_r.
 \label{eq:introduction-equality-pair}
\end{align}
For larger distance constraints, the endpoint
$\delta=1-D^{-1}$ remains admissible and attains the plateau $\log D$.
The limit
$\alpha\uparrow1$ recovers the sharp von Neumann conditional-entropy bound
of \cite{BertaEtAl2026}, i.e.~
\begin{align}
    \lim_{\alpha\uparrow1}
 \frac1{1-\alpha}
 \log\!\left[(1-\delta)^\alpha
 +(D-1)^{1-\alpha}\delta^\alpha\right]
 =h_2(\delta)+\delta\log(D-1),
\end{align}
where $h_2(\delta) := \delta \log \frac{1}{\delta} + (1-\delta) \log \frac{1}{1-\delta}$ is the binary entropy function.

The proof is guided by the isotropic states in \cref{eq:introduction-equality-pair} that saturate the bound.  The
$\alpha$-power of an extremal isotropic state splits into a maximally
entangled component and its orthogonal complement, whose partial traces
occur with multiplicities $1$ and $D-1$.  
For a general anchor state $\sigma_{\A\B}$,
Schmidt-rank domination \cite{terhal2000schmidt} produces a positive complement with exactly the
same partial-trace multiplicity.  Weighting the $\alpha$-power of the
anchor and this complement in the same proportions as in the isotropic
equality model yields an auxiliary comparison point whose R\'enyi
functional has precisely the value appearing in the claimed bound; see
\cref{fig:schmidt-envelope} for the schematic illustration.

We then linearize the relevant concave R\'enyi functional $\mathbb{Q}_\alpha(\rho)$ at this
comparison point $\tau_{\A\B}$:
\begin{align} \label{eq:concave_envelope}
\mathbb{Q}_\alpha(\rho)
 &\leq \Tr\!\left[\nabla\mathbb{Q}_\alpha(\tau)\rho_{\A\B}\right]
 =
 \underbrace{\Tr\!\left[\nabla\mathbb{Q}_\alpha(\tau)(\rho_{\A\B}-\sigma_{\A\B})\right]}_{\text{change-of-measure term}}
 +
\underbrace{\Tr\!\left[\nabla\mathbb{Q}_\alpha(\tau)\sigma_{\A\B}\right]}_{\text{anchor term}},
\end{align}
where $\nabla\mathbb{Q}_\alpha(\tau)$ is the Hilbert--Schmidt gradient operator at $\tau_{\A\B}$, as illustrated in \cref{fig:linearization}.
The variational trace-distance bound controls the change of the resulting
gradient operator from the anchor state $\sigma_{\A\B}$ to
$\rho_{\A\B}$, while a Fr\'echet-derivative noncommutative estimate controls the anchor term.
These two estimates reproduce the isotropic
equality calculation without degrading the sharp constant.
The sandwiched
proof uses the same comparison geometry; its gradient operator is
identified from the optimizer fixed-point equation \cite{HT14} and
Danskin's theorem \cite{Danskin1967}.

When the conditioning system is classical, the relevant operators commute
within every classical block.  The noncommutative estimates then reduce to
scalar power inequalities, so the Petz result extends to all
$0<\alpha<1$, recovering the Jabbour--Datta theorem
\cite{JabbourDatta2022}.

\begin{figure*}[htbp]\centering
	
	\begin{subfigure}[t]{0.48\textwidth}
		\centering
		\begin{tikzpicture}[>=stealth]
			\def\totalwidth{6.2}
			\def\statewidth{1.8} 
			\def\height{2.8}
			\colorlet{statebox}{gray!40}
			\colorlet{rembox}{gray!15}
			\colorlet{textdark}{black!80}
			
			\draw[decorate, decoration={calligraphic brace, amplitude=8pt}, thick, text=textdark]
			(0,\height+0.1) -- (\totalwidth,\height+0.1) 
			node[midway, above=10pt, align=center] {Isotropic envelope\\[1mm] $\I_{\A}\otimes P_{\B_0} \vphantom{\Tr_{\A}[\sigma_{\A\B}^\alpha]}$};
			
			\draw[thick, fill=statebox, text=textdark] (0,0) rectangle (\statewidth,\height);
			\node[align=center, text=textdark] at (\statewidth/2, \height/2) {
				Bell\\[1mm] sector\\[2mm] $\Phi_r$
			};
			
			\draw[thick, fill=rembox, text=textdark] (\statewidth,0) rectangle (\totalwidth,\height);
			\node[align=center, text=textdark] at ({(\totalwidth+\statewidth)/2}, \height/2) {
				Complement\\\\[2mm] $\I_{\A}\otimes P_{\B_0} - \Phi_r$
			};
			
			\draw[decorate, decoration={calligraphic brace, amplitude=6pt, mirror}, thick, text=textdark]
			(0,-0.2) -- (\statewidth,-0.2) 
			node[midway, below=8pt, align=center] {Partial trace\\[1mm]$1 \times \frac{1}{r}P_{\B_0}$};
			
			\draw[decorate, decoration={calligraphic brace, amplitude=6pt, mirror}, thick, text=textdark]
			(\statewidth,-0.2) -- (\totalwidth,-0.2) 
			node[midway, below=8pt, align=center] {Partial trace\\[1mm]$(D-1) \times \frac{1}{r}P_{\B_0}$};
		\end{tikzpicture}
		\caption{Isotropic equality model}
		\label{fig:envelope_saturation}
	\end{subfigure}\hfill
	\begin{subfigure}[t]{0.48\textwidth}
		\centering
		\begin{tikzpicture}[>=stealth]
			\def\totalwidth{6.2}
			\def\statewidth{1.8} 
			\def\height{2.8}
			\colorlet{statebox}{gray!40}
			\colorlet{rembox}{gray!15}
			\colorlet{textdark}{black!80}
			
			\draw[decorate, decoration={calligraphic brace, amplitude=8pt}, thick, text=textdark]
			(0,\height+0.1) -- (\totalwidth,\height+0.1) 
			node[midway, above=10pt, align=center] {Schmidt-rank envelope\\[1mm] $r\I_{\A}\otimes\Tr_{\A}[\sigma_{\A\B}^\alpha]$};
			
			\draw[thick, fill=statebox, text=textdark] (0,0) rectangle (\statewidth,\height);
			\node[align=center, text=textdark] at (\statewidth/2, \height/2) {
				Anchor\\[1mm] $\alpha$-power\\[2mm] $\sigma_{\A\B}^\alpha$
			};
			
			\draw[thick, fill=rembox, text=textdark] (\statewidth,0) rectangle (\totalwidth,\height);
			\node[align=center, text=textdark] at ({(\totalwidth+\statewidth)/2}, \height/2) {
				Positive remainder\\\\[2mm] $r\I_{\A}\otimes\Tr_{\A}[\sigma_{\A\B}^\alpha] - \sigma_{\A\B}^\alpha$
			};
			
			\draw[decorate, decoration={calligraphic brace, amplitude=6pt, mirror}, thick, text=textdark]
			(0,-0.2) -- (\statewidth,-0.2) 
			node[midway, below=8pt, align=center] {Partial trace\\[1mm]$1 \times \Tr_{\A}[\sigma_{\A\B}^\alpha] \vphantom{\frac{1}{r}}$};
			
			\draw[decorate, decoration={calligraphic brace, amplitude=6pt, mirror}, thick, text=textdark]
			(\statewidth,-0.2) -- (\totalwidth,-0.2) 
			node[midway, below=8pt, align=center] {Partial trace\\[1mm]$(D-1) \times \Tr_{\A}[\sigma_{\A\B}^\alpha] \vphantom{\frac{1}{r}}$};
		\end{tikzpicture}
		\caption{Arbitrary-anchor decomposition}
		\label{fig:envelope_generalization}
	\end{subfigure}
	
	\caption{
		Schematic correspondence between the isotropic equality model and the
		arbitrary-anchor construction. 
        Here, $P_{\B_0}$ denotes the projection onto the $r$-dimensional subspace of $\B$.
        In both panels, the darker positive
		operator and its lighter complement have partial traces (tracing out system $\A$) proportional to
		the same operator on $\B$, with proportionality factors $1$ and $D-1$.
		In panel~{(A)}
		the two summands have orthogonal supports, whereas in panel~{(B)}
		they need not commute or have orthogonal supports.
		The
		exact agreement of the partial-trace multiplicities motivates assigning
		the same weights $a$ and $b$ in
		\cref{eq:comparison-point-explicit}.
	}
	\label{fig:schmidt-envelope}
	
\end{figure*}
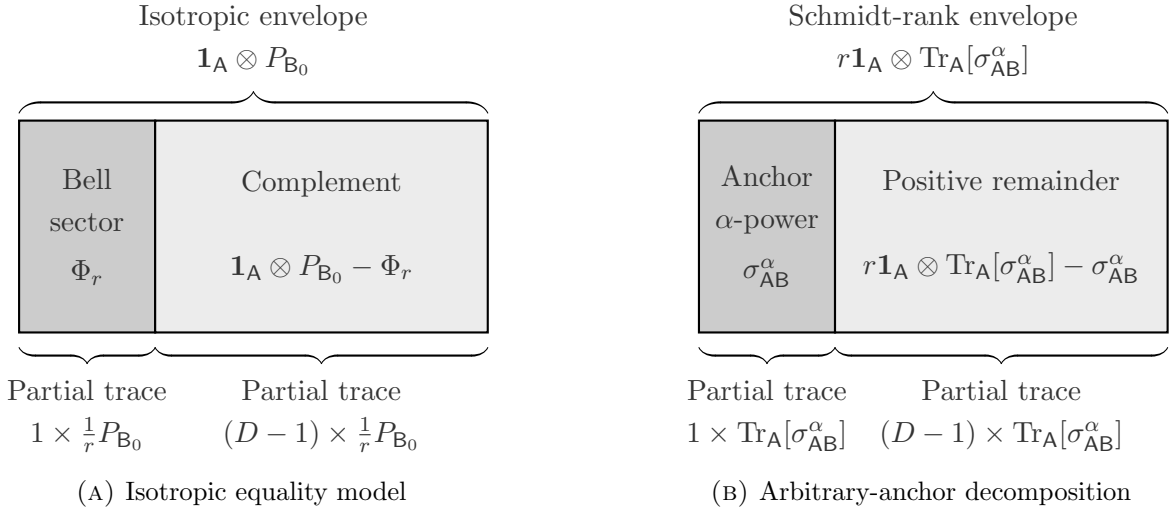

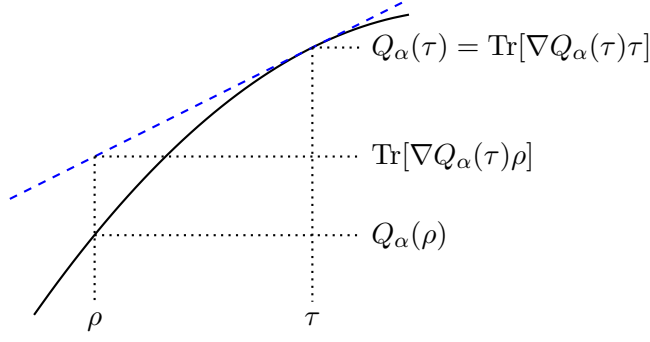
\begin{figure}[th]
	\centering
	\begin{tikzpicture}[scale=1.6, >=stealth]	
		\draw[thick, domain=1.7:4.8, smooth, variable=\x] 
		plot ({\x}, {-0.2*\x*\x + 2.1*\x - 2.2});
		
		\draw[thick, blue, dashed, domain=1.5:4.8, variable=\x] 
		plot ({\x}, {0.5*\x + 1});
		
		
		\coordinate (Tau) at (4, 3);
		\coordinate (TauX) at (4, 0.9);
		
		\coordinate (RhoCurve) at (2.2, 1.45);
		\coordinate (RhoTangent) at (2.2, 2.1);
		\coordinate (RhoX) at (2.2, 0.9);
		
		\draw[thick, dotted] (Tau) -- (TauX) node[below] {$\tau$};
		\draw[thick, dotted] (RhoTangent) -- (RhoX) node[below] {$\rho$};
		
		\draw[thick, dotted] (Tau) -- (4.4, 3) 
		node[right] {$Q_\alpha(\tau) = \mathrm{Tr}\!\left[\nabla Q_{\alpha}(\tau) \tau\right]$};
		
		\draw[thick, dotted] (RhoTangent) -- (4.4, 2.1) 
		node[right] {$\mathrm{Tr}\!\left[\nabla Q_{\alpha}(\tau) \rho\right]$};
		
		\draw[thick, dotted] (RhoCurve) -- (4.4, 1.45) 
		node[right] {$Q_\alpha(\rho)$};
		
	\end{tikzpicture}
	\caption{Linearization of the R\'enyi functional $\mathbb{Q}_{\alpha}(\rho)$ (defined in \cref{eq:Q-functional} and \cref{eq:Q}, respectively) 
	at the comparison point $\tau$ via the Hilbert--Schmidt gradient operator $\nabla Q_{\alpha}(\tau)$ defined in \cref{eq:general-gradient}.
	For the Bell anchor state $\Phi_r$, the chosen comparison point $\tau = \tau_{\Phi_r,\delta}^{(\alpha)}$ in \cref{eq:comparison-point-explicit} coincides with the worst-case isotropic state $\rho_\delta$ in \cref{eq:quantum-equality-family}.
	The gap between the tangent line and the R\'enyi functional curves thus collapses, i.e., $Q_{\alpha}(\rho_{\delta}) = \Tr[\nabla Q_{\alpha}(\tau)\rho_{\delta}] = Q_{\alpha}(\tau)$.
	} \label{fig:linearization}
\end{figure}

The paper is organized as follows.
\Cref{sec:preliminaries} states the main results and collects the elementary
ingredients.  \Cref{sec:quantum-model} uses the isotropic saturation model
to motivate the comparison point.
\Cref{sec:main-proof,sec:sand_main-proof} prove the Petz and sandwiched
continuity bounds, respectively.
\Cref{sec:classical-conditioning} treats a classical conditioning system.
\Cref{sec:discussion} provides discussions and outlook.
We leave the technical estimates to  \cref{sec:transport}.

\medskip
\emph{We would like to highlight again that this manuscript and its content were partly developed with the assistance
	of ChatGPT 5.6 Sol, a generative AI tool.
}

\section{Main results and elementary ingredients}\label{sec:preliminaries}

All Hilbert spaces are finite dimensional and all logarithms are natural.
We denote by $\mathcal{S}(\B)$ the state space on quantum system $\B$.
For
states $\rho$ and $\sigma$, their trace distance is
\begin{equation}\label{eq:trace-distance}
 T(\rho,\sigma):=\frac12\norm{\rho-\sigma}_1.
\end{equation}
Let
\begin{equation}
 r:=\min\{d_{\A},d_{\B}\},
 \qquad
 D:=d_{\A}r,
 \label{eq:rD}
\end{equation}
and define
\begin{equation}
 \Gamma_{\alpha,D}(\delta)
 :=\frac{1}{1-\alpha}
 \log\!\left[(1-\varepsilon)^{\alpha}+(D-1)^{1-\alpha}\varepsilon^{\alpha}\right],
 \qquad
 \varepsilon:=\min\{\delta,1-D^{-1}\}.
 \label{eq:Gamma}
\end{equation}

Let $\rho$ be a state and $\sigma\geq 0$ be a positive semidefinite operator.
For $0<\alpha<1$, the Petz--R\'enyi divergence \cite{Pet86} is
\begin{equation}\label{eq:Petz-divergence}
 D_\alpha(\rho\|\sigma)
 :=\frac{1}{\alpha-1}\log\Tr[\rho^\alpha\sigma^{1-\alpha}],
\end{equation}
and the (optimized) Petz--R\'enyi conditional entropy is
\begin{equation}\label{eq:conditional-entropy}
 H_\alpha^\uparrow(\A|\B)_\rho
 :=-\inf_{\sigma_{\B}\in\states{\B}}
 D_\alpha\!\left(\rho_{\A\B}\middle\|\I_{\A}\otimes\sigma_{\B}\right).
\end{equation}

For $X_{\A\B}\geq0$, define
\begin{equation}\label{eq:Q-functional}
 Q_\alpha(X_{\A\B})
 :=\Tr\!\left(\Tr_{\A}[X_{\A\B}^\alpha]\right)^{1/\alpha}.
\end{equation}
The quantum Sibson identity \cite[Lemma~3]{SharmaWarsi2013} gives
\begin{equation}\label{eq:Sibson-formula}
 H_\alpha^\uparrow(\A|\B)_\rho
 =\frac{\alpha}{1-\alpha}\log Q_\alpha(\rho).
\end{equation}

The functional $Q_\alpha$ is positively homogeneous of degree one and concave
on the positive cone for every $0<\alpha<1$.  This is a specialization of
Epstein's trace-concavity theorem \cite{Epstein1973}; see also the systematic
trace-functional treatment in \cite{CarlenFrankLieb2016}.

Our first main result is the following sharp continuity bound for the Petz--R\'enyi conditional entropy.

\begin{theorem}[Petz sharp continuity bound]\label{thm:main}
	Let $\alpha\in[\frac12,1)$ and let $\rho_{\A\B},\sigma_{\A\B}$ be states
	satisfying $T(\rho,\sigma)\leq\delta$.
	Let
	$D=d_{\A}\min\{d_{\A},d_{\B}\}$ be the effective dimension.  Then
	\begin{equation}\label{eq:main-bound-piecewise}
		\boxed{\;\;
			\abs{H_\alpha^\uparrow(\A|\B)_\rho
				-H_\alpha^\uparrow(\A|\B)_\sigma}
			\leq
			\Gamma_{\alpha,D}(\delta),
            \;\;}
	\end{equation}
    where $\Gamma_{\alpha,D}$ is defined in \cref{eq:Gamma}.
	The right-hand side is optimal for every distance constraint.
\end{theorem}

Let $\rho$ be a state and $\sigma\geq 0$ be a positive semidefinite operator.
For $\alpha\in[\frac12,1)$, the sandwiched R\'enyi divergence \cite{MDS+13, WWY14} is
\begin{equation}
 \sand D_{\alpha}(\rho\|\sigma)
 :=\frac{1}{\alpha-1}
 \log\Tr\!\left[\left(
 \sigma^{\frac{1-\alpha}{2\alpha}}
 \rho
 \sigma^{\frac{1-\alpha}{2\alpha}}
 \right)^{\alpha}\right],
 \label{eq:sand-div}
\end{equation}
with the standard support convention.  Its data-processing inequality holds
for every $\alpha\geq\frac12$
\cite{MDS+13,FL13,Bei13}.

For $X_{\A\B}\geq0$, define the sandwiched R\'enyi functional
\begin{equation}
 \sand Q_{\alpha}(X_{\A\B})
 :=
 \left(
 \sup_{\tilde{\sigma}_{\B}\in\states{\B}}
 \Tr\!\left[\left(
 (\I_{\A}\otimes\tilde{\sigma}_{\B}^{\frac{1-\alpha}{2\alpha}})
 X_{\A\B}
 (\I_{\A}\otimes\tilde{\sigma}_{\B}^{\frac{1-\alpha}{2\alpha}})
 \right)^{\alpha}\right]
 \right)^{1/\alpha}.
 \label{eq:Q}
\end{equation}
For a state $\rho_{\A\B}$, the sandwiched R\'enyi conditional entropy is
\begin{equation}
 \sand H_{\alpha}^{\uparrow}(\A|\B)_{\rho}
 :=\frac{\alpha}{1-\alpha}\log\sand Q_{\alpha}(\rho_{\A\B}).
 \label{eq:H-Q}
\end{equation}

\begin{theorem}[Sandwiched sharp continuity bound]\label{thm:sand_main}
	Let $\alpha\in[\frac12,1)$ and let $\rho_{\A\B},\sigma_{\A\B}$ be states
	satisfying $T(\rho,\sigma)\leq\delta$.
	Let
	$D=d_{\A}\min\{d_{\A},d_{\B}\}$ be the effective dimension.  Then
	\begin{equation}\label{eq:sand_main-bound-piecewise}
		\boxed{\;\;
			\abs{\sand H_\alpha^\uparrow(\A|\B)_\rho
				-\sand H_\alpha^\uparrow(\A|\B)_\sigma}
			\leq
			\Gamma_{\alpha,D}(\delta),
            \;\;}
	\end{equation}
    where $\Gamma_{\alpha,D}$ is defined in \cref{eq:Gamma}.
	The right-hand side is optimal for every distance constraint.
\end{theorem}

Taking $\alpha\uparrow1$ gives an alternative proof of the sharp order-one
bound.  Indeed, by letting $\varepsilon:=\min\{\delta,1-D^{-1}\}$, we have
\begin{equation}\label{eq:order-one-limit}
	\lim_{\alpha\uparrow1}
 \frac1{1-\alpha}
 \log\!\left[(1-\varepsilon)^\alpha
 +(D-1)^{1-\alpha}\varepsilon^\alpha\right]
 =h_2(\varepsilon)+\varepsilon\log(D-1).
\end{equation}
Moreover, for every state $\omega_{\A\B}$,
\begin{align}
    H_\alpha^\uparrow(\A|\B)_\omega
    &\to H(\A|\B)_{\omega} := - \Tr[ \omega_{\A\B} \log \omega_{\A\B}] + \Tr[ \omega_{\B} \log \omega_{\B}],
    \\
    \sand H_\alpha^\uparrow(\A|\B)_\omega
    &\to H(\A|\B)_{\omega}.
\end{align}
Hence
\begin{equation}\label{eq:order-one-bound}
	\abs{H(\A|\B)_\rho-H(\A|\B)_\sigma}
	\leq
	h_2(\varepsilon)+\varepsilon\log(D-1).
\end{equation}
This recovers the corresponding sharp order-one result of
\cite{BertaEtAl2026}.

\subsection{Elementary ingredients}

Let $0<\alpha<1$ and let $\tau_{\A\B}>0$. 
Let $\nabla Q_{\alpha}(\tau)$ be the Hilbert--Schmidt gradient of $Q_\alpha$ at $\tau$; i.e., $\nabla Q_{\alpha}(\tau)$ is a Hermitian operator such that 
\begin{align}\label{eq:general-gradient}
	\left.\frac{\dd}{\dd t} Q_{\alpha}(\tau + t H)\right|_{t=0}
	= \Tr\left[ \nabla Q_{\alpha}(\tau) H \right],
	\quad	\forall\, \tau>0, H=H^\dagger.
\end{align}
Since $\nabla Q_{\alpha}(\tau)$ draws a tangent plane that supports the concave Petz functional $Q_{\alpha}$ from above, we hence call $\nabla Q_{\alpha}(\tau)$ supporting operator and $\tau$ the comparison point.

We also denote by $\nabla\sand Q_{\alpha}(\tau)$ the supporting operator
of the sandwiched functional at $\tau$ similarly.

For $p>0$ and $X>0$, denote the Fr\'echet derivative of the power map by
\begin{equation}\label{eq:Frechet-derivative}
	\mathrm D_{p,X}(H)
	:=\left.\frac{\dd}{\dd t}(X+tH)^p\right|_{t=0}.
\end{equation}
It is Hilbert--Schmidt self-adjoint and invertible. For $p>1$, its inverse is positive,
because it is the derivative of the operator-monotone map
$Z\mapsto Z^{1/p}$ at $X^p$.  We use
\begin{equation}\label{eq:Frechet-identities}
	\mathrm D_{p,X}(\I)=pX^{p-1},
	\qquad
	\mathrm D_{p,X}(X)=pX^p,
	\qquad
	\mathrm D_{p,cX}=c^{p-1}\mathrm D_{p,X}
	\quad(c>0).
\end{equation}
These are standard facts from matrix-function calculus; see
\cite{Bhatia1997,Higham2008}.

\begin{lemma}[Supporting operator of $Q_\alpha$]\label{lem:supporting-operator}
Let $0<\alpha<1$ and let $\tau_{\A\B}>0$. 
The supporting operator defined in \cref{eq:general-gradient} satisfies
\begin{equation}\label{eq:general-gradient-formula}
\nabla Q_{\alpha}(\tau)
 =\frac1\alpha\,
 \mathrm D_{\alpha,\tau}
 \!\left[
 \I_{\A}\otimes
 \left(\Tr_{\A}[\tau_{\A\B}^{\alpha}]\right)^{\frac{1-\alpha}{\alpha}}
 \right].
\end{equation}
Moreover,
\begin{equation}\label{eq:supporting-score}
 Q_\alpha(X_{\A\B})\leq\Tr\!\left[\nabla Q_{\alpha}(\tau) X_{\A\B}\right]
 \quad(X_{\A\B}\geq0),
 \qquad
 Q_\alpha(\tau)=\Tr\!\left[ \nabla Q_{\alpha}(\tau) \tau\right].
\end{equation}
\end{lemma}

\begin{proof}
For every Hermitian \(H_{\A\B}\), differentiating
\(
Q_\alpha(X)=\Tr[(\Tr_{\A}[X^\alpha])^{1/\alpha}]
\)
at \(X=\tau\) gives
\begin{align}
 \mathrm DQ_\alpha(\tau)[H]
 &=
 \frac1\alpha
 \Tr\!\left[
 \left(\Tr_{\A}[\tau^\alpha]\right)^{\frac{1-\alpha}{\alpha}}
 \Tr_{\A}\!\left[\mathrm D_{\alpha,\tau}(H)\right]
 \right]
 \notag\\
 &=
 \frac1\alpha
 \Tr\!\left[
 \left(\I_{\A}\otimes
  \left(\Tr_{\A}[\tau^\alpha]\right)^{\frac{1-\alpha}{\alpha}}\right)
 \mathrm D_{\alpha,\tau}(H)
 \right]
 \notag\\
 &=
 \frac1\alpha
 \Tr\!\left[
 \mathrm D_{\alpha,\tau}\!\left(
 \I_{\A}\otimes
  \left(\Tr_{\A}[\tau^\alpha]\right)^{\frac{1-\alpha}{\alpha}}
 \right)H
 \right].
 \label{eq:gradient-calculation-general}
\end{align}
The third equality is from the Hilbert--Schmidt self-adjointness of
\(\mathrm D_{\alpha,\tau}\).
Since this holds for every Hermitian $H_{\A\B}$, we have $\nabla Q_{\alpha}(\tau)
 =\frac1\alpha\,
 \mathrm D_{\alpha,\tau}
 \!\left[
 \I_{\A}\otimes
 \left(\Tr_{\A}[\tau_{\A\B}^{\alpha}]\right)^{\frac{1-\alpha}{\alpha}}
 \right]$.
 
Concavity therefore yields
\[
 Q_\alpha(X)\leq Q_\alpha(\tau)+\Tr[\nabla Q_{\alpha}(\tau)(X-\tau)].
\]
Finally, Euler's theorem for the degree-one homogeneous functional $Q_\alpha$
gives
\[
 \Tr[\nabla Q_{\alpha}(\tau)\tau]
 =\left.\frac{\dd}{\dd t}Q_\alpha((1+t)\tau)\right|_{t=0}
 =Q_\alpha(\tau),
\]
which proves both statements in \cref{eq:supporting-score}.
\end{proof}

\begin{lemma}[Concavity of $\sand Q_{\alpha}$]
\label{lem:concavity}
Let $\alpha\in[\frac12,1)$.  The functional $\sand Q_{\alpha}$ in
\cref{eq:Q} is continuous, positively homogeneous of degree one, and
concave on the positive cone.
\end{lemma}

\begin{proof}
Homogeneity and continuity follow from \cref{eq:Q} and compactness of the
state space.  
The following direct-sum property of $\sand Q_{\alpha}$ was proved in  \cite{MDS+13, TBH14}:
\begin{equation}
 \sand Q_{\alpha}(X_0\oplus X_1)
 =\sand Q_{\alpha}(X_0)+\sand Q_{\alpha}(X_1).
 \label{eq:sand-direct-sum}
\end{equation}


Apply \cref{eq:sand-direct-sum} to
$\lambda X\oplus(1-\lambda)Y$ and discard the classical conditioning
register.  Sandwiched conditional-entropy data processing gives
\[
 \sand Q_{\alpha}(\lambda X+(1-\lambda)Y)
 \geq
 \lambda\sand Q_{\alpha}(X)
 +(1-\lambda)\sand Q_{\alpha}(Y),
\]
which proves concavity.
\end{proof}

\begin{lemma}[Sandwiched optimizer and supporting operator]
\label{lem:sand_supporting}
Let $\alpha\in[\frac12,1)$ and let $X_{\A\B}>0$.
There is a unique optimizer
$\sand\sigma_{\B}>0$ in \cref{eq:Q}.  
This optimizer is uniquely characterized by the fixed-point equation \cite{HT14, rubboli2026quantum}
\begin{align}
\sand\sigma_{\B}
=\frac{
\Tr_{\A}\!\left[\left(
(\I_{\A}\otimes\sand\sigma_{\B}^{\frac{1-\alpha}{2\alpha}})
X_{\A\B}
(\I_{\A}\otimes\sand\sigma_{\B}^{\frac{1-\alpha}{2\alpha}})
\right)^\alpha\right]}
{
\Tr\!\left[\left(
(\I_{\A}\otimes\sand\sigma_{\B}^{\frac{1-\alpha}{2\alpha}})
X_{\A\B}
(\I_{\A}\otimes\sand\sigma_{\B}^{\frac{1-\alpha}{2\alpha}})
\right)^\alpha\right]}.
\label{eq:fixed-point}
\end{align}
Moreover, $\sand Q_{\alpha}$ is differentiable at $X$ and
\begin{align}
&\nabla\sand Q_{\alpha}(X)
\notag\\[-1mm]
&=
\left(
\I_{\A}\otimes
\left\{
\Tr_{\A}\!\left[\left(
(\I_{\A}\otimes\sand\sigma_{\B}^{\frac{1-\alpha}{2\alpha}})
X_{\A\B}
(\I_{\A}\otimes\sand\sigma_{\B}^{\frac{1-\alpha}{2\alpha}})
\right)^\alpha\right]
\right\}^{\frac{1-\alpha}{2\alpha}}
\right)
\notag\\
&\quad\times
\left[
(\I_{\A}\otimes\sand\sigma_{\B}^{\frac{1-\alpha}{2\alpha}})
X
(\I_{\A}\otimes\sand\sigma_{\B}^{\frac{1-\alpha}{2\alpha}})
\right]^{-(1-\alpha)}
\notag\\
&\quad\times
\left(
\I_{\A}\otimes
\left\{
\Tr_{\A}\!\left[\left(
(\I_{\A}\otimes\sand\sigma_{\B}^{\frac{1-\alpha}{2\alpha}})
X_{\A\B}
(\I_{\A}\otimes\sand\sigma_{\B}^{\frac{1-\alpha}{2\alpha}})
\right)^\alpha\right]
\right\}^{\frac{1-\alpha}{2\alpha}}
\right).
\label{eq:gradient}
\end{align}
Consequently, for every $\tau_{\A\B}>0$ and $X_{\A\B}\geq0$,
\begin{equation}
 \sand Q_{\alpha}(X_{\A\B})
 \leq
 \Tr\!\left[\nabla\sand Q_{\alpha}(\tau)X_{\A\B}\right],
 \label{eq:sand-support}
\end{equation}
and equality holds at $X_{\A\B}=\tau_{\A\B}$.
\end{lemma}

\begin{proof}
Apply \cite[Lemma~5]{HT14} (see also \cite{rubboli2026quantum}) with the first reference operator equal to
$\I_{\A}$ and with $X/\Tr [X]$ in place of the normalized bipartite state.
Since multiplication of $X$ by a positive scalar does not change the
optimizer, that lemma gives uniqueness and shows that
\cref{eq:fixed-point} is both necessary and sufficient.  Because $X>0$,
its marginal $X_{\B}$ is faithful, and the optimizer is faithful as well.

For a fixed faithful state $\omega_{\B}$, write
\begin{align}
 Y_{X,\omega}
 :=
 (\I_{\A}\otimes
 \omega_{\B}^{\frac{1-\alpha}{2\alpha}})
 X
 (\I_{\A}\otimes
 \omega_{\B}^{\frac{1-\alpha}{2\alpha}}).
\end{align}
The fixed-$\omega_{\B}$ functional is differentiable in every Hermitian
direction $H$, with
\begin{align}
 &\mathrm D_X
 \left(\Tr[Y_{X,\omega}^{\alpha}]\right)^{1/\alpha}[H]
 =
 \left(\Tr[Y_{X,\omega}^{\alpha}]\right)^{
 \frac{1-\alpha}{\alpha}}
 \Tr\!\left[
 (\I_{\A}\otimes
 \omega_{\B}^{\frac{1-\alpha}{2\alpha}})
 Y_{X,\omega}^{\alpha-1}
 (\I_{\A}\otimes
 \omega_{\B}^{\frac{1-\alpha}{2\alpha}})
 H
 \right].
 \label{eq:fixed-reference-derivative}
\end{align}
The feasible set $\states{\B}$ is compact, the active optimizer is
unique and faithful, and the fixed-reference objective is continuously
Fr\'echet differentiable in a neighborhood of that optimizer.
Danskin's theorem \cite{Danskin1967} therefore gives, for every Hermitian $H$, its derivative is obtained by differentiating the fixed-\(\omega_B\) objective at \(\omega_B=\widetilde{\sigma}_B\):
\begin{align}
 \mathrm D\,\sand Q_\alpha(X)[H]
 =
 \left(\Tr[Y_{X,\sand\sigma}^{\alpha}]\right)^{
 \frac{1-\alpha}{\alpha}}
 \cdot
 \Tr\!\left[
 (\I_{\A}\otimes
 \sand\sigma_{\B}^{\frac{1-\alpha}{2\alpha}})
 Y_{X,\sand\sigma}^{\alpha-1}
 (\I_{\A}\otimes
 \sand\sigma_{\B}^{\frac{1-\alpha}{2\alpha}})
 H
 \right].
 \label{eq:Danskin-derivative}
\end{align}
Thus the operator multiplying $H$ is the gradient of
$\sand Q_\alpha$ at $X$.  Substituting the fixed-point identity
\cref{eq:fixed-point} into \cref{eq:Danskin-derivative}, the scalar
normalization factors cancel, and one obtains exactly
\cref{eq:gradient}.

Finally, concavity and positive homogeneity of degree one of $\sand Q_\alpha$ gives \cref{eq:sand-support}
for every $\tau>0$ as in the proof of \cref{lem:supporting-operator}.
\end{proof}

\begin{remark}[Non-faithful states] \label{remark:sand_support}
For general $X_{\A\B}\geq0$, let $P_{\B}=\supp (X_{\B})$.  Then
$X=(\I_{\A}\otimes P_{\B})X(\I_{\A}\otimes P_{\B})$, and
\cite[Lemma~5]{HT14} gives
$\supp(\sand \sigma_{\B})=P_{\B}$.  The optimizer theorem may therefore be
applied on the reduced conditioning space supported on $P_{\B}$.  
In the main
continuity proof it is cleaner to keep $X_{\A\B}>0$ and pass to singular states by
the common regularization at the end of the theorem.
\end{remark}

\begin{lemma}[Variational trace-distance bound]\label{lemma:trace-distance}
For states $\rho,\sigma$ and a Hermitian operator $G$,
\begin{equation}\label{eq:quotient-duality}
 \abs{\Tr[G(\rho-\sigma)]}
 \leq
 T(\rho,\sigma)\,
 \bigl(\lambda_{\max}(G)-\lambda_{\min}(G)\bigr).
\end{equation}
Moreover,
\begin{equation}\label{eq:quotient-norm}
 \frac{\lambda_{\max}(G)-\lambda_{\min}(G)}{2}
 =\inf_{c\in\mathbb R}\norm{G-c\I}_\infty.
\end{equation}
\end{lemma}

\begin{proof}
Since $\Tr[(\rho-\sigma)]=0$, for every $c\in\mathbb R$,
\[
 \Tr[G(\rho-\sigma)]
 =\Tr[(G-c\I)(\rho-\sigma)].
\]
Trace-norm duality therefore gives
\[
 \abs{\Tr[G(\rho-\sigma)]}
 \leq
 \norm{G-c\I}_\infty\norm{\rho-\sigma}_1.
\]
The quantity $\norm{G-c\I}_\infty = \max_i | \lambda_i(G) - c |$ is minimized when $c$ is the midpoint of
the spectral interval of $G$, and its minimum is
$\frac12(\lambda_{\max}(G)-\lambda_{\min}(G))$ because the minimum of $\max_i | \lambda_i(G) - c |  \geq \max\{|\lambda_{\max}-c|,|\lambda_{\min}-c|\}$ is achieved by $c = \frac12(\lambda_{\max}(G)+\lambda_{\min}(G))$.
Since
$\norm{\rho-\sigma}_1=2T(\rho,\sigma)$, both claims follow.
\end{proof}


\begin{lemma}[Schmidt-rank domination~{\cite{terhal2000schmidt}}]\label{lem:Schmidt-domination}
For every positive operator $X_{\A\B}$,
\begin{equation}\label{eq:Schmidt-domination}
 X_{\A\B}
 \leq
 r\I_{\A}\otimes\Tr_{\A}[X_{\A\B}],
 \qquad
 r:=\min\{d_{\A},d_{\B}\}.
\end{equation}
\end{lemma}

\begin{proof}
It suffices to consider $X=|\psi\rangle\!\langle\psi|$.  Write a Schmidt
decomposition
$|\psi\rangle=\sum_{j=1}^{k}\sqrt{\lambda_j}|a_jb_j\rangle$, where
$k\leq r$.  For every vector $|v\rangle$, Cauchy--Schwarz gives
\[
 \abs{\langle\psi|v\rangle}^2
 \leq k\sum_{j=1}^{k}\lambda_j
       \abs{\langle a_jb_j|v\rangle}^2
 \leq r\langle v|\I_{\A}\otimes\Tr_{\A}[X_{\A\B}]|v\rangle.
\]
Summing a spectral decomposition proves the general case.
\end{proof}

\begin{lemma}[Range] \label{lem:lower_bound}
    Let $r:=\min\{d_{\A},d_{\B}\}$.
	For any state $\omega_{\A\B}$ and $\alpha \in [\frac12,1)$,
	\begin{align}
		 &-\log r
		\leq
		H_{\alpha}^{\uparrow}(\A|\B)_{\omega}
		\leq
		\log d_{\A},
		\qquad
		Q_\alpha(\omega)\geq r^{-\frac{1-\alpha}{\alpha}},
		 \label{eq:Q-lower-range}
        \\
        &-\log r
		\leq
		\sand H_{\alpha}^{\uparrow}(\A|\B)_{\omega}
		\leq
		\log d_{\A},
		\qquad
		\sand Q_\alpha(\omega)\geq r^{-\frac{1-\alpha}{\alpha}}.
	\end{align}
\end{lemma}
\begin{proof}
Indeed, \cref{lem:Schmidt-domination} and operator monotonicity of
$x^{\frac{1-\alpha}{\alpha}}$ for $0\leq \frac{1-\alpha}{\alpha}\leq1$, give
\begin{align}
	1=\Tr[\omega_{\A\B}]
	&=\Tr\!\left[
	\omega_{\A\B}^\alpha
	(\omega_{\A\B}^\alpha)^{\frac{1-\alpha}{\alpha}}
	\right]\notag\\
	&\leq
	\Tr\!\left[
	\omega_{\A\B}^\alpha
	\left(
	r\I_{\A}\otimes\Tr_{\A}[\omega_{\A\B}^\alpha]
	\right)^{\frac{1-\alpha}{\alpha}}
	\right]
	=r^{\frac{1-\alpha}{\alpha}}Q_\alpha(\omega).
	\label{eq:Q-lower-range-proof}
\end{align}
The lower bound for $\sand Q_{\alpha}$ follows from the monotonicity of the sandwiched
divergence in its order relative to the max-relative entropy.
For the upper
bound, data processing under $\Tr_{\A}$ gives, for every $\sigma_{\B}$,
\begin{align}
 \sand D_{\alpha}
 (\omega_{\A\B}\|\I_{\A}\otimes\sigma_{\B})
 \geq
 \sand D_{\alpha}(\omega_{\B}\|d_{\A}\sigma_{\B})
 \geq-\log d_{\A},
\end{align}
and the upper bound for $H_{\alpha}^{\uparrow}(\A|\B)_{\omega}$ follows similarly.
\end{proof}

\section{Saturation of the continuity bound}
\label{sec:quantum-model}

We first examine the quantum states that saturate the continuity bound in \cref{eq:main-bound-piecewise} and identify the features responsible for saturation.  
This motivates us to preserve those features when constructing a comparison point for an
arbitrary anchor state $\sigma_{\A\B}$.  
The continuity estimate for a general pair
$(\rho,\sigma)$ will be proved later in \cref{sec:main-proof}.

\subsection{The isotropic equality model}

Choose $r$-dimensional subspaces $\A_0\subseteq\A$ and
$\B_0\subseteq\B$, let $\Phi_r$ be a maximally entangled state on
$\A_0\otimes\B_0$, and let $P_{\B_0}$ denote the projection onto
$\B_0$.  Recall that
\[
D=d_{\A}r.
\]
The maximally entangled sector and its orthogonal complement inside
$\A\otimes\B_0$ satisfy
\begin{equation}\label{eq:quantum-equality-multiplicity}
	\Tr_{\A}[\Phi_r]=\frac1rP_{\B_0},
	\qquad
	\Tr_{\A}\!\left[
	\I_{\A}\otimes P_{\B_0}-\Phi_r
	\right]
	=\frac{D-1}{r}P_{\B_0}.
\end{equation}
Thus their partial traces are proportional to the same operator
$P_{\B_0}/r$, with proportionality factors $1$ and $D-1$.

For $0\leq\delta\leq1-D^{-1}$, consider the isotropic family
\begin{equation}\label{eq:quantum-equality-family}
	\rho_\delta
	:=
	(1-\delta)\Phi_r
	+\frac{\delta}{D-1}
	\left(
	\I_{\A}\otimes P_{\B_0}-\Phi_r
	\right).
\end{equation}
The positive and negative parts of $\rho_\delta-\Phi_r$ are supported on
the two orthogonal sectors in
\cref{eq:quantum-equality-multiplicity}, and both have trace $\delta$.
Hence
\begin{equation}\label{eq:isotropic-distance}
	T(\rho_\delta,\Phi_r)=\delta.
\end{equation}
Since the two sectors are orthogonal,
\begin{equation}\label{eq:root-equality-family}
	\rho_\delta^\alpha
	=
	a\Phi_r
	+b\left(
	\I_{\A}\otimes P_{\B_0}-\Phi_r
	\right),
	\qquad
	\boxed{
		a:=(1-\delta)^\alpha,
		\quad
		b:=\left(\frac{\delta}{D-1}\right)^\alpha.}
\end{equation}
Thus $a$ and $b$ are exactly the two nonzero spectral values of
$\rho_\delta^\alpha$ on $\A\otimes\B_0$, with multiplicities $1$ and
$D-1$, respectively.  The increasing branch
$\delta\leq1-D^{-1}$ is precisely the condition $a\geq b$.

Taking the partial trace in \cref{eq:root-equality-family} gives
\begin{equation}\label{eq:isotropic-root-marginal}
	\Tr_{\A}[\rho_\delta^\alpha]
	=
	\frac{a+(D-1)b}{r}P_{\B_0}.
\end{equation}
Consequently,
\begin{equation}\label{eq:isotropic-Q-values}
	Q_\alpha(\Phi_r)
	=
	r^{1-1/\alpha},
	\qquad
	Q_\alpha(\rho_\delta)
	=
	r^{1-1/\alpha}
	\bigl[a+(D-1)b\bigr]^{1/\alpha}.
\end{equation}
The Sibson formula therefore gives the saturation identity
\begin{align}
	H_\alpha^\uparrow(\A|\B)_{\rho_\delta}
	-H_\alpha^\uparrow(\A|\B)_{\Phi_r}
	&=
	\frac{1}{1-\alpha}
	\log\bigl[a+(D-1)b\bigr]
	\notag\\
	&=
	\frac{1}{1-\alpha}
	\log\!\left[
	(1-\delta)^\alpha
	+(D-1)^{1-\alpha}\delta^\alpha
	\right].
	\label{eq:isotropic-saturation}
\end{align}

Similarly,\footnote{
Indeed, a straightforward calculation by \cref{lem:sand_supporting} and \cref{remark:sand_support} shows that the optimizer to $\sand Q_{\alpha}(\omega)$ for either $\omega\in\{\rho_{\delta},\Phi_r\}$ is given by $\sand \sigma_{\B} = P_{\B_0}/r$. 
Hence, $\sand Q_{\alpha}(\rho_\delta)
 =r^{-\frac{1-\alpha}{\alpha}}
 [a+(D-1)b]^{1/\alpha}$,
 and
 $\sand Q_{\alpha}(\Phi_r)
 =r^{-\frac{1-\alpha}{\alpha}}$.
}
\begin{align}
	\sand H_\alpha^\uparrow(\A|\B)_{\rho_\delta}
	- \sand H_\alpha^\uparrow(\A|\B)_{\Phi_r}
	&=
	\frac{1}{1-\alpha}
	\log\bigl[a+(D-1)b\bigr]
	\notag\\
	&=
	\frac{1}{1-\alpha}
	\log\!\left[
	(1-\delta)^\alpha
	+(D-1)^{1-\alpha}\delta^\alpha
	\right].
\end{align}

The equality case also indicates where the supporting operator should be
evaluated.  To see this, suppose
$0<\delta<1-D^{-1}$ and restrict attention to the effective support
$\A\otimes\B_0$.  A direct blockwise evaluation of the Petz supporting operator gives
\begin{align}
	\nabla Q_{\alpha}(\rho_\delta)
	&=
	\left(
	\frac{a+(D-1)b}{r}
	\right)^{\frac{1-\alpha}{\alpha}}
	\left[
	a^{-\frac{1-\alpha}{\alpha}}\Phi_r
	+b^{-\frac{1-\alpha}{\alpha}}
	\left(
	\I_{\A}\otimes P_{\B_0}-\Phi_r
	\right)
	\right].
	\label{eq:isotropic-supporting-operator}
\end{align}
Since $a\geq b$, the smaller
eigenvalue of $\nabla Q_{\alpha}(\rho_\delta)$ occurs on the maximally entangled sector
and the larger one occurs on its complement.  Moreover,
$\rho_\delta-\Phi_r$ removes exactly $\delta$ units of trace from the
first sector and places them in the second.  It follows that
\begin{align}
	Q_\alpha(\rho_\delta)
	&=
	\Tr[\nabla Q_{\alpha}(\rho_\delta)\rho_\delta]
	\notag\\
	&=
	\Tr[\nabla Q_{\alpha}(\rho_\delta)\Phi_r]
	+\delta\left(
	\lambda_{\max}(\nabla Q_{\alpha}(\rho_\delta))
	-\lambda_{\min}(\nabla Q_{\alpha}(\rho_\delta))
	\right).
	\label{eq:isotropic-supporting-saturation}
\end{align}
Here the spectral extrema are taken on the effective support
$\A\otimes\B_0$.

Hence, in the equality model, the supporting operator at the boundary state
$\rho_\delta$ reproduces the sharp value exactly: its expectation on the
anchor $\Phi_r$, together with the maximal variation permitted by trace
distance $\delta$, equals $Q_\alpha(\rho_\delta)$.  This observation
suggests that, for a general anchor, we should construct a comparison point
that retains the same weights $a,b$ and the same partial-trace
multiplicities.

\subsection{Transferring the equality geometry to an arbitrary anchor}

Now fix an arbitrary anchor $\sigma_{\A\B}$.  The state maximizing
$Q_\alpha$ in the trace-distance ball around $\sigma_{\A\B}$ is not known
explicitly, so there is no available analogue of the extremal state
$\rho_\delta$.  Instead, we construct a comparison point by preserving the
features of $\rho_\delta^\alpha$ identified above.

Applying \cref{lem:Schmidt-domination} to
$\sigma_{\A\B}^{\alpha}$ gives
\begin{equation}\label{eq:root-envelope-explicit}
	\sigma_{\A\B}^{\alpha}
	\leq
	r\I_{\A}\otimes
	\Tr_{\A}[\sigma_{\A\B}^{\alpha}],
\end{equation}
and hence
\begin{equation}\label{eq:root-complement-explicit}
	r\I_{\A}\otimes
	\Tr_{\A}[\sigma_{\A\B}^{\alpha}]
	-
	\sigma_{\A\B}^{\alpha}
	\geq0.
\end{equation}
The two positive summands
\[
\sigma_{\A\B}^{\alpha}
\quad\text{and}\quad
r\I_{\A}\otimes
\Tr_{\A}[\sigma_{\A\B}^{\alpha}]
-
\sigma_{\A\B}^{\alpha}
\]
need not commute or have orthogonal supports.  Nevertheless, their partial
traces satisfy exactly the same $1:(D-1)$ relation as in the equality
model:
\begin{equation}\label{eq:arbitrary-anchor-multiplicity}
	\Tr_{\A}\!\left[
	r\I_{\A}\otimes
	\Tr_{\A}[\sigma_{\A\B}^{\alpha}]
	-
	\sigma_{\A\B}^{\alpha}
	\right]
	=
	(D-1)\Tr_{\A}[\sigma_{\A\B}^{\alpha}].
\end{equation}


Motivated by \cref{eq:root-equality-family}, we retain the same weights
$a$ and $b$ and define
\begin{equation}\label{eq:comparison-point-explicit}
	\boxed{
		\tau_{\sigma,\delta}^{(\alpha)}
		:=
		\left[
		a\sigma_{\A\B}^{\alpha}
		+b\left(
		r\I_{\A}\otimes
		\Tr_{\A}[\sigma_{\A\B}^{\alpha}]
		-
		\sigma_{\A\B}^{\alpha}
		\right)
		\right]^{1/\alpha}.}
\end{equation}
This comparison point need not be normalized; only its positivity is
required for the supporting-hyperplane argument.  Its $\alpha$-power has
the same weights $a,b$ as the isotropic equality state, while
\cref{eq:arbitrary-anchor-multiplicity} preserves the same partial-trace
multiplicities.  
The correspondence is summarized schematically in
\cref{fig:schmidt-envelope}.
Consequently,
\begin{equation}\label{eq:comparison-partial-trace}
	\Tr_{\A}\!\left[
	(\tau_{\sigma,\delta}^{(\alpha)})^\alpha
	\right]
	=
	\bigl[a+(D-1)b\bigr]
	\Tr_{\A}[\sigma_{\A\B}^{\alpha}],
\end{equation}
and therefore
\begin{equation}\label{eq:comparison-target-ratio}
	Q_\alpha(\tau_{\sigma,\delta}^{(\alpha)})
	=
	\bigl[a+(D-1)b\bigr]^{1/\alpha}
	Q_\alpha(\sigma).
\end{equation}

By the Sibson formula, the desired one-sided continuity estimate is
equivalent to
\begin{equation}\label{eq:target-Q-ratio}
	Q_\alpha(\rho)
	\leq
	\bigl[a+(D-1)b\bigr]^{1/\alpha}
	Q_\alpha(\sigma).
\end{equation}
Thus the right-hand side of \cref{eq:target-Q-ratio} is exactly
$Q_\alpha(\tau_{\sigma,\delta}^{(\alpha)})$.  This equality fixes the
desired value of the comparison point, but it does not yet prove that this
value bounds $Q_\alpha(\rho)$.

The additional motivation for evaluating the supporting operator at
$\tau_{\sigma,\delta}^{(\alpha)}$ comes from
\cref{eq:isotropic-supporting-saturation}.  When
$\sigma_{\A\B}=\Phi_r$, the construction gives
\begin{equation}\label{eq:comparison-recovers-isotropic}
	r\I_{\A}\otimes\Tr_{\A}[\Phi_r^\alpha]
	=
	\I_{\A}\otimes P_{\B_0},
	\qquad
	\tau_{\Phi_r,\delta}^{(\alpha)}
	=
	\rho_\delta,
\end{equation}
and its supporting operator satisfies the equality
\cref{eq:isotropic-supporting-saturation}.  We are therefore led to ask
whether, for a general anchor state $\sigma_{\A\B}$, the supporting operator $\nabla Q_{\alpha}$ at
$\tau_{\sigma,\delta}^{(\alpha)}$, denoted by $G_{\sigma,\delta}^{(\alpha)} := \nabla Q_{\alpha}(\tau_{\sigma,\delta}^{(\alpha)})$ satisfies the corresponding inequality
\begin{align}
	&\Tr\!\left[
	G_{\sigma,\delta}^{(\alpha)}\sigma_{\A\B}
	\right]
	+\delta\left(
	\lambda_{\max}
	\left(G_{\sigma,\delta}^{(\alpha)}\right)
	-
	\lambda_{\min}
	\left(G_{\sigma,\delta}^{(\alpha)}\right)
	\right)
	\leq
	Q_\alpha(\tau_{\sigma,\delta}^{(\alpha)}).
	\label{eq:comparison-supporting-goal}
\end{align}
This is precisely what will be proved in \cref{sec:main-proof}.  Once
\cref{eq:comparison-supporting-goal} is established, the supporting
inequality for $Q_\alpha$ (\cref{lem:supporting-operator}) and the variational trace-distance bound (\cref{lemma:trace-distance}) yield
\cref{eq:target-Q-ratio}.

The construction therefore mimics the isotropic equality state in the
specific features relevant to the later proof: positivity of the
complement, the weights $a,b$, the $1:(D-1)$ partial-trace relation, the
target value of $Q_\alpha$, and the point at which the supporting operator
is evaluated. 

Following the similar reasoning, define the sandwiched comparison point
$\sand\tau_{\A\B}$ by
\begin{align}
&\left(
(\I_{\A}\otimes\sand\sigma_{\B}^{\frac{1-\alpha}{2\alpha}})
\sand\tau_{\A\B}
(\I_{\A}\otimes\sand\sigma_{\B}^{\frac{1-\alpha}{2\alpha}})
\right)^{\alpha}
\notag\\
&=(a-b)
\left(
(\I_{\A}\otimes\sand\sigma_{\B}^{\frac{1-\alpha}{2\alpha}})
\sigma_{\A\B}
(\I_{\A}\otimes\sand\sigma_{\B}^{\frac{1-\alpha}{2\alpha}})
\right)^{\alpha}
+
rb\,\I_{\A}\otimes
\Tr_{\A}\!\left[\left(
(\I_{\A}\otimes\sand\sigma_{\B}^{\frac{1-\alpha}{2\alpha}})
\sigma_{\A\B}
(\I_{\A}\otimes\sand\sigma_{\B}^{\frac{1-\alpha}{2\alpha}})
\right)^{\alpha}\right],
\label{eq:comparison}
\end{align}
where
$\sand\sigma_{\B}$ is an optimizer in \cref{eq:Q}.
The first coefficient is nonnegative exactly on the increasing branch
$\delta\leq1-D^{-1}$.  The operator $\sand\tau$ need not be normalized.
This causes no difficulty because $\sand Q_{\alpha}$ is homogeneous of
degree one and its supporting operator is invariant under positive rescaling
of the comparison point.

Taking the partial trace of \cref{eq:comparison} immediately gives
\begin{align}
&\Tr_{\A}\!\left[\left(
(\I_{\A}\otimes\sand\sigma_{\B}^{\frac{1-\alpha}{2\alpha}})
\sand\tau_{\A\B}
(\I_{\A}\otimes\sand\sigma_{\B}^{\frac{1-\alpha}{2\alpha}})
\right)^{\alpha}\right]
=
[a+(D-1)b]\,\Tr_{\A}\!\left[\left(
(\I_{\A}\otimes\sand\sigma_{\B}^{\frac{1-\alpha}{2\alpha}})
\sigma_{\A\B}
(\I_{\A}\otimes\sand\sigma_{\B}^{\frac{1-\alpha}{2\alpha}})
\right)^{\alpha}\right].
 \label{eq:comparison-marginal}
\end{align}
This identity explains why the same optimizer $\sand\sigma_{\B}$ remains
optimal at the comparison point and why the comparison value is exactly the
sharp target value.

\section{Proof of the Petz sharp continuity bound}\label{sec:main-proof}

Fix $\alpha\in[\frac12,1)$ 
We first assume that $\sigma_{\A\B}>0$ and
$0<\delta<1-D^{-1}$, so $a>b>0$.
Let
\begin{equation}\label{eq:specialized-supporting-operator}
 G_{\sigma,\delta}^{(\alpha)}
 := \nabla Q_{\alpha}(\tau_{\sigma,\delta}^{(\alpha)})
\end{equation}
be the supporting operator defined in \eqref{eq:general-gradient} at the
comparison point $\tau_{\sigma,\delta}^{(\alpha)}$ given in \cref{eq:comparison-point-explicit}.  
Since
$\mathrm D_{\alpha,\tau}
 =(\mathrm D_{1/\alpha,\tau^\alpha})^{-1}$, the analytical formula
\cref{eq:general-gradient-formula} in \cref{lem:supporting-operator}, together with
\cref{eq:comparison-partial-trace}, gives the explicit expression
\begin{align}
 G_{\sigma,\delta}^{(\alpha)}
 & =
 \left(\frac{a+(D-1)b}{r}\right)^{\frac{1-\alpha}{\alpha}}
 \frac1\alpha
 \left(
 \mathrm D_{\frac1\alpha,\,
 a\sigma_{\A\B}^\alpha
 +b\left(
 r\I_{\A}\otimes\Tr_{\A}[\sigma_{\A\B}^\alpha]
 -\sigma_{\A\B}^\alpha
 \right)}
 \right)^{-1}\left[
 \left(
 r\I_{\A}\otimes\Tr_{\A}[\sigma_{\A\B}^\alpha]
 \right)^{\frac{1-\alpha}{\alpha}}
 \right].
 \label{eq:explicit-support-score}
\end{align}
Here we also used
\[
 \left(
 r\I_{\A}\otimes\Tr_{\A}[\sigma_{\A\B}^\alpha]
 \right)^{\frac{1-\alpha}{\alpha}}
 =r^{\frac{1-\alpha}{\alpha}}\I_{\A}\otimes
 \left(\Tr_{\A}[\sigma_{\A\B}^\alpha]\right)^{\frac{1-\alpha}{\alpha}}.
\]
%

\cref{prop:score-calibration} below provides crucial estimates for the supporting operator $G_{\sigma,\delta}^{(\alpha)}$ that will be used in our proof shortly.

\begin{proposition}[Bounds on Petz supporting operator]
\label{prop:score-calibration}
The supporting operator in \cref{eq:explicit-support-score} satisfies
\begin{align}
 &\text{(spectral bounds)}
 &&\left(\frac{a+(D-1)b}{r}\right)^{\frac{1-\alpha}{\alpha}}
 a^{-\frac{1-\alpha}{\alpha}}\I_{\A\B}
 \leq G_{\sigma,\delta}^{(\alpha)}
 \leq
 \left(\frac{a+(D-1)b}{r}\right)^{\frac{1-\alpha}{\alpha}}
 b^{-\frac{1-\alpha}{\alpha}}\I_{\A\B},
 \label{eq:explicit-score-interval}
\\
&\text{(anchor calibration)}
&&\Tr\!\left[G_{\sigma,\delta}^{(\alpha)}\sigma_{\A\B}\right]
\leq
 \left(\frac{a+(D-1)b}{a}\right)^{\frac{1-\alpha}{\alpha}}
 Q_\alpha(\sigma). \label{eq:anchor-calibration}
\end{align}
Consequently,
\begin{align}
 &\lambda_{\max}\!\left(G_{\sigma,\delta}^{(\alpha)}\right)
 -\lambda_{\min}\!\left(G_{\sigma,\delta}^{(\alpha)}\right)
 \leq
 \left(\frac{a+(D-1)b}{r}\right)^{\frac{1-\alpha}{\alpha}}
 \left(b^{-\frac{1-\alpha}{\alpha}}-a^{-\frac{1-\alpha}{\alpha}}\right).
 \label{eq:explicit-score-diameter}
\end{align}
\end{proposition}
\noindent The proof is deferred to \cref{sec:Petz_technical}.


\medskip
We are now ready to prove the main result of the sharp continuity bound of \cref{thm:main}.

\begin{proof}[Proof of \cref{thm:main}]
Let $T(\rho,\sigma)\leq\delta$.  By
\cref{lem:supporting-operator} with the supporting operator given in
\cref{eq:specialized-supporting-operator}, we have
\begin{align}\label{eq:supporting-transport-split}
 Q_\alpha(\rho)
 &\leq \Tr\!\left[G_{\sigma,\delta}^{(\alpha)}\rho\right]
 =
 \underbrace{\Tr\!\left[G_{\sigma,\delta}^{(\alpha)}(\rho-\sigma)\right]}_{\text{change-of-measure term}}
+ \underbrace{\Tr\!\left[G_{\sigma,\delta}^{(\alpha)}\sigma\right]}_{\text{anchor term}}.
\end{align}
Applying \cref{lemma:trace-distance} together with
\cref{eq:explicit-score-diameter,eq:anchor-calibration} in \cref{prop:score-calibration} yields
\begin{align}
 Q_\alpha(\rho)
 &\leq
 \left(\frac{a+(D-1)b}{a}\right)^{\frac{1-\alpha}{\alpha}}Q_\alpha(\sigma)
 +r^{-\frac{1-\alpha}{\alpha}}\delta
 \bigl[a+(D-1)b\bigr]^{\frac{1-\alpha}{\alpha}}
 \bigl(b^{-\frac{1-\alpha}{\alpha}}-a^{-\frac{1-\alpha}{\alpha}}\bigr).
 \label{eq:score-proof-before-scalar}
\end{align}

Next, we rewrite the second term of \cref{eq:score-proof-before-scalar}.
By \cref{eq:root-equality-family}, the weights satisfy
\begin{equation}\label{eq:root-mass-identities}
 a^{1/\alpha}=1-\delta,
 \qquad
 (D-1)b^{1/\alpha}=\delta.
\end{equation}
Since $1+\frac{1-\alpha}{\alpha}=\frac{1}{\alpha}$, we can expand the sum of the weights as
\begin{equation*}
	a+(D-1)b 
	= (1-\delta)a^{-\frac{1-\alpha}{\alpha}} + \delta b^{-\frac{1-\alpha}{\alpha}} 
	= a^{-\frac{1-\alpha}{\alpha}} + \delta\bigl(b^{-\frac{1-\alpha}{\alpha}} - a^{-\frac{1-\alpha}{\alpha}}\bigr).
\end{equation*}
Multiplying both sides by $[a+(D-1)b]^{\frac{1-\alpha}{\alpha}}$ and rearranging yields an exact identity for the coefficient in \cref{eq:score-proof-before-scalar}:
\begin{equation}\label{eq:exact-scalar-closure}
	\delta\bigl[a+(D-1)b\bigr]^{\frac{1-\alpha}{\alpha}}
	\bigl(b^{-\frac{1-\alpha}{\alpha}}-a^{-\frac{1-\alpha}{\alpha}}\bigr)
	=
	\bigl[a+(D-1)b\bigr]^{1/\alpha}
	-\left(\frac{a+(D-1)b}{a}\right)^{\frac{1-\alpha}{\alpha}}.
\end{equation}
Substitute \cref{eq:exact-scalar-closure} into
\cref{eq:score-proof-before-scalar}.  Since
$Q_\alpha(\sigma)\geq r^{-\frac{1-\alpha}{\alpha}}$ by
\cref{lem:lower_bound}, the intermediate terms cancel:
\begin{align}
 Q_\alpha(\rho)
 &\leq
 \left(\frac{a+(D-1)b}{a}\right)^{\frac{1-\alpha}{\alpha}}Q_\alpha(\sigma)\notag\\
 &\quad+
 r^{-\frac{1-\alpha}{\alpha}}\left[
 \bigl[a+(D-1)b\bigr]^{1/\alpha}
 -\left(\frac{a+(D-1)b}{a}\right)^{\frac{1-\alpha}{\alpha}}
 \right]\notag\\
 &\leq
 \bigl[a+(D-1)b\bigr]^{1/\alpha}Q_\alpha(\sigma).
 \label{eq:Q-ratio-final}
\end{align}
By \cref{eq:Sibson-formula},
\begin{align}
 H_\alpha^\uparrow(\A|\B)_\rho
 -H_\alpha^\uparrow(\A|\B)_\sigma
 &\leq
 \frac{1}{1-\alpha}\log\bigl[a+(D-1)b\bigr]\notag\\
 &=\frac{1}{1-\alpha}
 \log\!\left[(1-\delta)^\alpha
 +(D-1)^{1-\alpha}\delta^\alpha\right].
 \label{eq:one-sided-final}
\end{align}
Interchanging $\rho$ and $\sigma$ proves the absolute-value bound for full-rank
states.
For general states, set
\begin{align}
\rho_\epsilon
:=
(1-\epsilon)\rho
+\epsilon\frac{\I_{\A\B}}{d_{\A}d_{\B}},
\qquad
\sigma_\epsilon
:=
(1-\epsilon)\sigma
+\epsilon\frac{\I_{\A\B}}{d_{\A}d_{\B}}.
\end{align}
Then $T(\rho_\epsilon,\sigma_\epsilon)
=(1-\epsilon)T(\rho,\sigma)
\leq\delta$.
Apply the full-rank result and let $\epsilon\downarrow0$, using the
continuity of $Q_\alpha$.

At
$\delta=1-D^{-1}$, its value is $\log D$.  For
$1/\alpha \geq 1$, the Schatten partial-trace estimate
\[
\norm{\Tr_{\A}[X]}_{1/\alpha}
\leq d_{\A}^{1-\alpha}\norm{X}_{1/\alpha}
\]
applied to $X=\rho^\alpha$ gives
\begin{equation}\label{eq:partial-trace-Schatten}
	Q_\alpha(\rho)\leq d_{\A}^{\frac{1-\alpha}{\alpha}},
\end{equation}
while \cref{eq:Q-lower-range} gives
$Q_\alpha(\rho)\geq r^{-\frac{1-\alpha}{\alpha}}$.  By
\cref{eq:Sibson-formula},
\begin{equation}\label{eq:entropy-range}
	-\log r
	\leq H_\alpha^\uparrow(\A|\B)_\rho
	\leq\log d_{\A}.
\end{equation}
Thus the total range is $\log(d_{\A}r)=\log D$, proving the plateau.

Sharpness on the increasing branch ($0\leq\delta<1-D^{-1}$) was established directly in
\cref{eq:isotropic-distance,eq:isotropic-saturation}.  At and after the
transition point, the pair $(\rho_{1-1/D},\Phi_r)$ saturates the continuity bound under
the constraint $T(\rho_{\A\B},\sigma_{\A\B})\leq\delta$.
\end{proof}


\section{Proof of the sandwiched sharp continuity bound}\label{sec:sand_main-proof}

Assume that $\sigma_{\A\B}>0$ and
$0<\delta\leq1-D^{-1}$, and let $\sand\sigma_{\B}$ be its optimizer.
Define the comparison point $\sand\tau_{\A\B}$ by \cref{eq:comparison}.

\begin{proposition}[Bounds on the sandwiched supporting operator]
\label{prop:sand_calibrated}
Fix an anchor state $\sigma_{\A\B} >0$.
The state $\sand\sigma_{\B}$ is the unique optimizer for
$\sand\tau_{\A\B}$, and
\begin{equation}
 \sand Q_{\alpha}(\sand\tau)
 =[a+(D-1)b]^{1/\alpha}\sand Q_{\alpha}(\sigma).
 \label{eq:comparison-value}
\end{equation}
Moreover,
\begin{align}
&\text{(spectral bounds)}
&&\left[\frac{a+(D-1)b}{ra}\right]^{\frac{1-\alpha}{\alpha}}
 \I_{\A\B}
\leq\nabla\sand Q_{\alpha}(\sand\tau)
\leq
\left[\frac{a+(D-1)b}{rb}\right]^{\frac{1-\alpha}{\alpha}}
 \I_{\A\B},
\label{eq:score-window}
\\
&\text{(anchor calibration)}
&&\Tr\!\left[\nabla\sand Q_{\alpha}(\sand\tau)\sigma\right]
\leq
\left[\frac{a+(D-1)b}{a}\right]^{\frac{1-\alpha}{\alpha}}
\sand Q_{\alpha}(\sigma).
\label{eq:sand-anchor-calibration}
\end{align}
\end{proposition}
\noindent The proof is deferred to \cref{sec:sand_technical}.

\begin{proof}[Proof of \cref{thm:sand_main}]
Assume first that $\rho$ and $\sigma$ are faithful and that
$0<\delta\leq1-D^{-1}$ with $T(\rho,\sigma)\leq\delta$.  Use $\sigma$ as the
anchor, choose its optimizer $\sand\sigma_{\B}$, and construct the comparison
point $\sand\tau$ by \cref{eq:comparison}.  
\cref{lem:sand_supporting} gives
\[
 \sand Q_{\alpha}(\rho)
 \leq
 \Tr\!\left[\nabla\sand Q_{\alpha}(\sand\tau)\rho\right].
\]
By \cref{lemma:trace-distance,prop:sand_calibrated},
\begin{align}
\sand Q_{\alpha}(\rho)
&\leq
\left[
\frac{a+(D-1)b}{a}
\right]^{\frac{1-\alpha}{\alpha}}
\sand Q_{\alpha}(\sigma)
+
\delta
\left[
\frac{a+(D-1)b}{r}
\right]^{\frac{1-\alpha}{\alpha}}
\times
\left[
b^{-\frac{1-\alpha}{\alpha}}
-a^{-\frac{1-\alpha}{\alpha}}
\right].
 \label{eq:before-scalar}
\end{align}
The scalar identity
\begin{align}
&\delta[a+(D-1)b]^{\frac{1-\alpha}{\alpha}}
\left[
 b^{-\frac{1-\alpha}{\alpha}}
 -a^{-\frac{1-\alpha}{\alpha}}
\right]
=
[a+(D-1)b]^{1/\alpha}
-
\left[
\frac{a+(D-1)b}{a}
\right]^{\frac{1-\alpha}{\alpha}}
\label{eq:scalar-identity}
\end{align}
follows from the barycentric identity
\begin{equation}
 a+(D-1)b
 =(1-\delta)a^{-\frac{1-\alpha}{\alpha}}
 +\delta b^{-\frac{1-\alpha}{\alpha}}.
\label{eq:barycentric}
\end{equation}
The right-hand side of \cref{eq:scalar-identity} is nonnegative on the
increasing branch.  Using
$\sand Q_{\alpha}(\sigma)\geq r^{-\frac{1-\alpha}{\alpha}}$ from
\cref{lem:lower_bound} in \cref{eq:before-scalar}, we obtain
\begin{align}
 \sand Q_{\alpha}(\rho)
 \leq
 [a+(D-1)b]^{1/\alpha}
 \sand Q_{\alpha}(\sigma).
 \label{eq:Q-ratio}
\end{align}
Equation \cref{eq:H-Q} now gives the required one-sided entropy estimate.
Interchanging $\rho$ and $\sigma$ proves the absolute-value bound.

At $\delta=1-D^{-1}$, one has
$\Gamma_{\alpha,D}(\delta)=\log D$,
which is the full entropy range from \cref{lem:lower_bound}; the same plateau holds
for every larger distance.  The case $\delta=0$ is immediate.

For singular states, use the common regularization
\begin{align}
 \rho_{\epsilon}=(1-\epsilon)\rho+\epsilon\frac{\I_{\A\B}}{d_{\A}d_{\B}},
 \qquad
 \sigma_{\epsilon}=(1-\epsilon)\sigma+\epsilon\frac{\I_{\A\B}}{d_{\A}d_{\B}}.
\end{align}
The trace distance does not increase, the faithful theorem applies, and
continuity from \cref{lem:concavity} permits $\epsilon\downarrow0$.
\end{proof}

\section{Classical conditioning: recovery of Jabbour--Datta}
\label{sec:classical-conditioning}

We finally consider quantum-classical states with a classical conditioning system $\Y$,
\begin{equation}\label{eq:quantum-classical-state}
	\rho_{\A\Y}
	=\sum_y\rho_{\A}^{y}\otimes|y\rangle\!\langle y|_{\Y},
\end{equation}
where the positive blocks $\rho_{\A}^{y}$ are subnormalized.  
This is the
quantum-classical setting of Jabbour and Datta
\cite{JabbourDatta2022}: the unconditioned system $\A$ may be fully quantum,
but the side information $\Y$ is classical.

\begin{corollary}[Jabbour--Datta continuity bound]
	\label{cor:Jabbour-Datta}
Let $0<\alpha<1$, and let $\rho_{\A\Y},\sigma_{\A\Y}$ be states of the form
\cref{eq:quantum-classical-state} in the same classical basis.  If
$T(\rho,\sigma)\leq\delta$, then
\begin{equation}\label{eq:Jabbour-Datta-bound}
 \abs{H_\alpha^\uparrow(\A|\Y)_\rho
 -H_\alpha^\uparrow(\A|\Y)_\sigma}
 \leq
 \begin{cases}
 \displaystyle
 \frac{1}{1-\alpha}
 \log\!\left[(1-\delta)^\alpha
 +(d_{\A}-1)^{1-\alpha}\delta^\alpha\right],
 &0\leq\delta\leq1-d_{\A}^{-1},\\[3mm]
 \log d_{\A},
 &1-d_{\A}^{-1}\leq\delta\leq1.
 \end{cases}
\end{equation}
The bound is sharp.
\end{corollary}

\begin{proof}
	For a classical conditioning system,
	\begin{align}
		\begin{split}
		\sigma_{\A\Y}^{\alpha}
		&=
		\sum_y(\sigma_{\A}^{y})^\alpha
		\otimes|y\rangle\!\langle y|_{\Y},
		\\
		\I_{\A}\otimes\Tr_{\A}[\sigma_{\A\Y}^{\alpha}]
		&=
		\sum_y
		\Tr[(\sigma_{\A}^{y})^\alpha]\I_{\A}
		\otimes|y\rangle\!\langle y|_{\Y}.
		\label{eq:classical-blocks}
		\end{split}
	\end{align}
	Hence,
	\[
	0\leq\sigma_{\A\Y}^{\alpha}
	\leq
	\I_{\A}\otimes\Tr_{\A}[\sigma_{\A\Y}^{\alpha}],
	\qquad
	\left[
	\sigma_{\A\Y}^{\alpha},
	\I_{\A}\otimes\Tr_{\A}[\sigma_{\A\Y}^{\alpha}]
	\right]=0.
	\]
	Thus the fully quantum construction applies with \(r=1\) and
	\(D=d_{\A}\).  More importantly, within each classical \(y\)-block the
	conditional envelope is a scalar multiple of \(\I_{\A}\).  Consequently,
	the inverse Fr\'echet derivative in the supporting operator reduces to
	ordinary scalar division.  By \cref{remark:commuting}, both the spectral bounds and the anchor calibration in
	\cref{prop:score-calibration} therefore hold for every
	\(p=1/\alpha>1\), rather than only for \(1<p\leq2\) because  neither operator convexity of
	$x^{1/\alpha}$ nor operator monotonicity of
	$x^{(1-\alpha)/\alpha}$ is required.
	
	The remaining lower-range estimate is also immediate blockwise:
	\[
	Q_\alpha(\sigma_{\A\Y})
	=
	\sum_y
	\left(\Tr[(\sigma_{\A}^{y})^\alpha]\right)^{1/\alpha}
	\geq
	\sum_y\Tr[\sigma_{\A}^{y}]
	=1,
	\]
	because
	\(
	(\Tr [Z])^{1/\alpha}\geq\Tr[Z^{1/\alpha}]
	\)
	for \(Z\geq0\) and \(1/\alpha>1\).
	Therefore the proof from
	\cref{eq:supporting-transport-split} through
	\cref{eq:Q-ratio-final} applies verbatim with \(r=1\) and
	\(D=d_{\A}\), for every \(0<\alpha<1\).
	Interchanging \(\rho\) and
	\(\sigma\), followed by a common block-diagonal full-rank
	regularization, gives the absolute-value bound.  
	The plateau follows from
    $1\leq Q_\alpha(\omega_{\A\Y})
    \leq d_{\A}^{(1-\alpha)/\alpha}$.
    Sharpness already holds for states supported on one classical value of
    $\Y$ and diagonal in a fixed basis of $\A$.
\end{proof}


\section{Discussion and outlook}\label{sec:discussion}

\subsection{Continuity of exponent functions}
The established continuity theorem also yields quantitative robustness guarantees for
the exponent functions arising in the quantum information-processing tasks
discussed in \cref{sec:introduction}.  As an illustration, consider the
following random-coding exponent function at rate \(R\):
\begin{align}
    E_{\mathrm{r}}(R)_{\rho}
    := \sup_{\alpha \in [1/2,1)}  \frac{1-\alpha}{\alpha} \left[ R - H_{\alpha}^{\uparrow}(\A|\B)_{\rho} \right].
\end{align}
Then, \cref{thm:main} implies that
\begin{align}
    E_{\mathrm{r}}(R)_{\rho}
    &\leq \sup_{\alpha \in [1/2,1)} \left\{ \frac{1-\alpha}{\alpha} \left[ R - H_{\alpha}^{\uparrow}(\A|\B)_{\sigma} \right] + \frac{1-\alpha}{\alpha} \Gamma_{\alpha,D}(\delta) \right\}
    \\
    &\leq E_{\mathrm{r}}(R)_{\sigma} + \sup_{\alpha \in [1/2,1)} \frac{1-\alpha}{\alpha} \Gamma_{\alpha,D}(\delta)
    \\
    &= E_{\mathrm{r}}(R)_{\sigma} + \Gamma_{1/2,D}(\delta).
\end{align}
Then, the worst sensitivity of the random-coding exponent is governed by the endpoint $\alpha = 1/2$, i.e.,
\begin{align}
    \Gamma_{1/2,D}(\delta)
    &= 2 \log \left[ \sqrt{1-\varepsilon} + \sqrt{(D-1)\varepsilon} \right],
    \qquad \varepsilon := \min\{\delta, 1-D^{-1}\}
    \\
    &\leq
    2\sqrt{(D-1)\delta}
\end{align}
Hence, the random-coding exponent is uniformly $1/2$-H\"older continuous in trace distance.

\subsection{Relation to the order-one proof}
Our argument and the sharp order-one proof of \cite{BertaEtAl2026} share a similar intrinsic structure.
Indeed, the comparison point in
\cref{eq:comparison-point-explicit} converges $(\alpha\to1)$ to the normalized state
\begin{equation} \label{eq:order-one-tau}
    \tau_{\sigma,\delta}^{(1)}
    =
    (1-\delta)\sigma_{\A\B}
    +
    \frac{\delta}{D-1}
    \left(
        r\I_{\A}\otimes\sigma_{\B}
        -
        \sigma_{\A\B}
    \right).
\end{equation}
The partial-trace identity
\cref{eq:comparison-partial-trace} correspondingly $\Tr_{\A}[\tau_{\sigma,\delta}^{(1)}]
    =
    \sigma_{\B}$.
The first-order limit of our supporting operator
\(G_{\sigma,\delta}^{(\alpha)}\) converges to
\begin{equation}
\label{eq:order-one-score-limit}
    \frac{
        G_{\sigma,\delta}^{(\alpha)}
        -
        \I_{\A\B}
    }{1-\alpha}
    \rightarrow
    G_{\A\B}
    :=
    -\log\tau_{\sigma,\delta}^{(1)}
    +
    \I_{\A}\otimes\log\sigma_{\B}.
\end{equation}
The supporting-operator bound (\Cref{lem:supporting-operator}) implies that
\begin{align}
    H(\A|\B)_\rho
    &\leq
    \Tr\!\left[
        G_{\A\B} \rho_{\A\B}
    \right]
    =
    \underbrace{\Tr\!\left[
        G_{\A\B} (\rho_{\A\B}-\sigma_{\A\B})
    \right]}_{\text{change-of-measure term}}
    +
    \underbrace{\Tr\!\left[
        G_{\A\B} \sigma_{\A\B}
    \right]}_{\text{anchor term}},
    \label{eq:order-one-score-decomposition}
\end{align}
where the inequality coincides with the relative-entropy data-processing inequality in \cite{BertaEtAl2026}:
\begin{equation}
    \Tr\!\left[
         G_{\A\B} \rho_{\A\B}
    \right]
    -
    H(\A|\B)_\rho
    =
    D\!\left(
        \rho_{\A\B}\middle\|\tau_{\sigma,\delta}^{(1)}
    \right)
    -
    D\!\left(
        \rho_{\B}\middle\|\sigma_{\B}
    \right)
    \geq0.
\end{equation}

The spectral bound, \(-\log(ar) \I_{\A\B} \leq G_{\sigma,\delta}^{(1)}\leq -\log (br) \I_{\A\B}\) follows directly from \cref{eq:order-one-tau} and \cref{eq:root-equality-family} as shown in \cite{BertaEtAl2026}, 
while the anchor term is exactly
\begin{align}
    \Tr\!\left[
        G_{\A\B} \sigma_{\A\B}
    \right]
    = H(\A|\B)_{\sigma} + D\!\left(
        \sigma_{\A\B}\big\|\tau_{\sigma,\delta}^{(1)}
    \right).
\end{align}
Since 
$D\!\left(
        \sigma_{\A\B}\big\|\tau_{\sigma,\delta}^{(1)}
    \right) \leq -\log (1-\delta)$ by \cref{eq:order-one-tau}, putting these estimates together gives
\begin{align}
    H(\A|\B)_\rho-H(\A|\B)_\sigma
    &\leq
    \delta\log
    \frac{(D-1)(1-\delta)}{\delta}
    -\log(1-\delta)
    \notag\\
    &=
    h_2(\delta)+\delta\log(D-1).
\end{align}
Thus, at order one, data processing establishes the supporting-operator
inequality and the Schmidt-rank domination
\(\tau_{\sigma,\delta}^{(1)}\geq(1-\delta)\sigma\) provides the anchor
calibration.  
In our R\'enyi proof, these roles are played respectively by
the concavity of \(Q_\alpha\) and the calibration estimate in
\cref{prop:score-calibration}, though the latter part and the spectral bounds of the supporting operator
require more involved matrix analysis for the case of $\alpha \in [\frac12,1)$.

\subsection{Orders below one half}
The fully quantum proof relies on the facts that
\(x^{1/\alpha}\) is operator convex and that
\(x^{(1-\alpha)/\alpha}\) is both operator monotone and operator concave.
These properties hold in the required ranges precisely when
\(\alpha\in[1/2,1)\), and are used in
\cref{lem:common-shift,thm:one-ray-transport}.
For \(\alpha<1/2\), the corresponding powers leave these operator ranges,
so the present proof does not extend.  
Whether the same sharp continuity
modulus remains valid for \(0<\alpha<1/2\) is left open.

\subsection{Orders above one}
The binary modulus \(\Gamma_{\alpha,D}\) from \cref{eq:Gamma} does not
extend sharply to the min-entropy endpoint. 
To see this, let us recall the folklore result of the sharp continuity bound of the min-entropy:\footnote{Indeed, let
\(\sand Q_{\infty}(\omega):=\inf\{\Tr [X_{\mathsf B}]:X_{\mathsf B}\geq0,\,
\omega_{\mathsf A\mathsf B}\leq \I_{\mathsf A}\otimes X_{\mathsf B}\}\),
so that
\(\widetilde H_\infty^\uparrow(\mathsf A|\mathsf B)_\omega=-\log \sand Q_{\infty}(\omega)\)
and \(d_{\mathsf A}^{-1}\leq \sand Q_{\infty}(\omega)\leq r\).
Writing \(T(\rho,\sigma)\leq\delta\) and
\(\Delta_+=(\rho-\sigma)_+\), the domination
\(\Delta_+\leq rI_{\mathsf A}\otimes\Tr_{\mathsf A}[\Delta_+]\)
implies \(\sand Q_{\infty}(\rho)\leq \sand Q_{\infty}(\sigma)+r\delta\).
After interchanging \(\rho,\sigma\), this gives
\(\max\{\sand Q_{\infty}(\rho)/\sand Q_{\infty}(\sigma),\sand Q_{\infty}(\sigma)/\sand Q_{\infty}(\rho)\}
 \leq1+d_{\mathsf A}r\delta=1+D\delta\);
the range of \(\sand Q_{\infty}\) additionally gives the plateau \(\log D\).
Sharpness follows from the isotropic family
\eqref{eq:quantum-equality-family}: for every
\(\delta\in[0,1]\), the pair 
$\left(
\rho_{\max\{1-1/D-\delta,\,0\}},
\rho_{1-1/D}
\right)$
attains the corresponding right-hand side.
}
\begin{align}
\abs{\sand H_{\infty}^\uparrow(\A|\B)_\rho
				-\sand H_{\infty}^\uparrow(\A|\B)_\sigma}
			\leq
            \min\{\log D, \log(1+D\delta)\}.
\end{align}
For background and basic properties on conditional min-entropy, see
\cite{tom-thesis}; related, almost sharp, continuity bounds are
given in \cite{BeigiGoodarzi2023, BluhmCapelGondolfMoebus2026}.
By contrast, the formal endpoint of the binary expression satisfies,
for \(0<\delta<1-D^{-1}\),
\begin{equation}
    \Gamma_{\infty,D}(\delta)
    =
    \log\frac1{1-\delta}
    <
    \min\left\{
        \log D,\,
        \log(1+D\delta)
    \right\}.
\end{equation}
Hence the sharp extremal geometry changes above order one: the continuation
of the \(\alpha<1\) binary modulus strictly underestimates the sensitivity
of the conditional min-entropy.  To the best of our knowledge, for finite
orders \(1<\alpha<\infty\), the sharp continuity modulus remains unknown
even for the classical R\'enyi conditional entropy.

\subsection{Other conditional R\'enyi entropies}
Sharp unrestricted continuity bounds remain open for the non-optimized
Petz and sandwiched conditional R\'enyi entropies, in which the reference
state on the conditioning system is fixed to the marginal
\(\rho_{\B}\) rather than optimized.
Related continuity inequalities for the non-optimized sandwiched
conditional R\'enyi entropy, under the additional assumption that the two
states have the same \(\B\)-marginal, were recently obtained by
Vershynina \cite{Vershynina26}.
Another direction is to study conditional entropies induced by other
quantum R\'enyi divergences, including the \((\alpha,z)\)-family
\cite{AD15}, the integral-representation and layer-cake divergences
\cite{hirche2023quantum,LHC25_layer_cake}, and the
\((\alpha,z,\lambda)\)-family \cite{rubboli2026quantum}.

\section*{Acknowledgements}
We sincerely thank Mario Berta for insightful discussions and useful suggestions, and Milan Mosonyi for organizing the \emph{\href{https://math.bme.hu/~mosonyi/QIMP2026/}{Quantum Information Theory and Mathematical Physics 2026}} workshop in Budapest, where this work was completed.

\appendix

\section{Technical Lemmas}\label{sec:transport}

We prove \cref{prop:score-calibration}.  
The following theorem is the noncommutative ingredient for handling the conditioning quantum system $\B$.

\subsection{A common-shift contraction}

For $1<p\leq2$, define the Bregman divergence gap as
\begin{equation}\label{eq:Bregman-gap}
 \Breg_p(A,B)
 :=\Tr\!\left[A^p-B^p-pB^{p-1}(A-B)\right].
\end{equation}

\begin{lemma}[Common-shift contraction]\label{lem:common-shift}
Let $1<p\leq2$, $A>0$, and $C,E\geq0$.  Then
\begin{equation}\label{eq:common-shift}
 \Breg_p(A+E,A+E+C)
 \leq
 \Breg_p(A,A+C).
\end{equation}
\end{lemma}

\begin{proof}
Set $A_t:=A+tE$ and $B_t:=A+C+tE$.  Differentiating and using
self-adjointness of Fr\'echet derivatives gives
\begin{align}
 \frac{\dd}{\dd t}\Breg_p(A_t,B_t)
 =p\Tr E\!\left[
 A_t^{p-1}-B_t^{p-1}
 +\mathrm D_{p-1,B_t}(C)
 \right].
 \label{eq:Bregman-derivative}
\end{align}
Because $0<p-1\leq1$, the map $X\mapsto X^{p-1}$ is operator concave.  Its
tangent inequality at $B_t$ gives
\[
 A_t^{p-1}
 =(B_t-C)^{p-1}
 \leq
 B_t^{p-1}-\mathrm D_{p-1,B_t}(C).
\]
Hence the derivative in \cref{eq:Bregman-derivative} is nonpositive.
Integration from $t=0$ to $t=1$ proves the claim.
\end{proof}

\subsection{The one-ray transport theorem} \label{sec:Petz_technical}

\begin{theorem}[One-ray transport inequality]\label{thm:one-ray-transport}
Let $1<p\leq2$, let $0<X\leq Z$, and let $0<\theta\leq1$.  Define
\begin{equation}\label{eq:transported-score}
 Y_\theta
 :=p\,\mathrm D_{p,\,X+\theta(Z-X)}^{-1}(Z^{p-1}).
\end{equation}
Then
\begin{equation}\label{eq:transport-score-interval}
 \I= Y_1 \leq Y_\theta\leq\theta^{1-p}\I
\end{equation}
and
\begin{equation}\label{eq:transport-anchor}
 \Tr[X^pY_\theta]\leq\Tr[XZ^{p-1}].
\end{equation}
\end{theorem}

\begin{remark}[Commuting operators] \label{remark:commuting}
	In the commuting case $[X,Z] = 0$, $Y_\theta =  \frac{Z^{p-1}}{[X+\theta(Z-X)]^{p-1}}$.
	\cref{eq:transport-score-interval} holds for all $p>1$ by bounding the denominator as $[\theta X + \theta(Z-X)]^{p-1} \leq [X+\theta(Z-X)]^{p-1} \leq Z^{p-1}$.
	\cref{eq:transport-anchor} holds for all $p>1$ as well 
	because $X^p Y_\theta \leq X^p Y_0 = X Z^{p-1}$.
\end{remark}

\begin{proof}
Since
\[
 \theta Z\leq X+\theta(Z-X)\leq Z
\]
and $x^{p-1}$ is operator monotone,
\[
 \theta^{p-1}Z^{p-1}
 \leq
 \bigl[X+\theta(Z-X)\bigr]^{p-1}
 \leq
 Z^{p-1}.
\]
Compare these inequalities with
\[
 \mathrm D_{p,\,X+\theta(Z-X)}(\I)
 =p\bigl[X+\theta(Z-X)\bigr]^{p-1}
\]
and apply the positive inverse derivative.  This proves
\cref{eq:transport-score-interval}.

For the anchor estimate, define the positive tangent defect
\begin{align}
 \Delta_\theta
 &:={}
 \theta\,\mathrm D_{p,\,X+\theta(Z-X)}(Z-X)\notag\\
 &\quad-
 \left(\bigl[X+\theta(Z-X)\bigr]^p-X^p\right).
 \label{eq:tangent-defect}
\end{align}
Its positivity follows from operator convexity of $x^p$.  Taking traces,
homogeneity, and \cref{lem:common-shift} give
\begin{align}
 \Tr\Delta_\theta
 &=\theta^p
 \Breg_p\!\left(\theta^{-1}X,
 \theta^{-1}X+Z-X\right)\notag\\
 &\leq\theta^p\Breg_p(X,Z).
 \label{eq:defect-common-shift}
\end{align}
Moreover,
\begin{align}
 &(p-1)\Tr[(Z-X)Z^{p-1}]-\Breg_p(X,Z)\notag\\
 &\qquad=\Tr\!\left[X(Z^{p-1}-X^{p-1})\right]\geq0,
 \label{eq:Bregman-budget}
\end{align}
because $x^{p-1}$ is operator monotone.  Thus
\begin{equation}\label{eq:defect-trace-bound}
 \Tr\Delta_\theta
 \leq(p-1)\theta^p\Tr[(Z-X)Z^{p-1}].
\end{equation}
Using $Y_\theta\leq\theta^{1-p}\I$ and
$\Delta_\theta\geq0$,
\begin{equation}\label{eq:transported-defect-bound}
 \Tr[Y_\theta\Delta_\theta]
 \leq(p-1)\theta\Tr[(Z-X)Z^{p-1}].
\end{equation}
Self-adjointness and
$\mathrm D_{p,\,X+\theta(Z-X)}(Y_\theta)=pZ^{p-1}$ now yield
\begin{align}
 &\Tr\!\left[Y_\theta
 \left(\bigl[X+\theta(Z-X)\bigr]^p-X^p\right)\right]\notag\\
 &\qquad=
 p\theta\Tr[(Z-X)Z^{p-1}]
 -\Tr[Y_\theta\Delta_\theta]\notag\\
 &\qquad\geq
 \theta\Tr[(Z-X)Z^{p-1}].
 \label{eq:secant-lower-bound}
\end{align}
Finally, linearity and \cref{eq:Frechet-identities} imply
\begin{align}
 \Tr[XZ^{p-1}]
 &=\frac1p\Tr\!\left[
 Y_\theta\,\mathrm D_{p,\,X+\theta(Z-X)}(X)
 \right]\notag\\
 &=\Tr\!\left[Y_\theta
 \bigl[X+\theta(Z-X)\bigr]^p\right]
 -\theta\Tr[(Z-X)Z^{p-1}].
 \label{eq:anchor-rewrite}
\end{align}
Subtracting $\Tr[X^pY_\theta]$ and using
\cref{eq:secant-lower-bound} proves \cref{eq:transport-anchor}.
\end{proof}

\begin{proof}[Proof of \cref{prop:score-calibration}]
Apply \cref{thm:one-ray-transport} with
\[
 p=\frac1\alpha,
 \qquad
 X=\sigma_{\A\B}^\alpha,
 \qquad
 Z=r\I_{\A}\otimes\Tr_{\A}[\sigma_{\A\B}^\alpha],
 \qquad
 \theta=\frac ba.
\]
The $\alpha$-root of the comparison point in
\cref{eq:comparison-point-explicit} is $a[X+\theta(Z-X)]$.  Homogeneity of the Fr\'echet derivative therefore turns
\cref{eq:transport-score-interval} into
\begin{align}
 a^{-\frac{1-\alpha}{\alpha}}\I
 &\leq
 \frac1\alpha
 \left(
 \mathrm D_{\frac1\alpha,\,
 a\sigma_{\A\B}^\alpha
 +b\left(
 r\I_{\A}\otimes\Tr_{\A}[\sigma_{\A\B}^\alpha]
 -\sigma_{\A\B}^\alpha
 \right)}
 \right)^{-1}\notag\\[-1mm]
 &\hspace{25mm}\left[
 \left(
 r\I_{\A}\otimes\Tr_{\A}[\sigma_{\A\B}^\alpha]
 \right)^{\frac{1-\alpha}{\alpha}}
 \right]
 \leq b^{-\frac{1-\alpha}{\alpha}}\I.
 \label{eq:applied-score-interval}
\end{align}
Multiplication by the prefactor in
\cref{eq:explicit-support-score} proves
\cref{eq:explicit-score-interval}.

Similarly, \cref{eq:transport-anchor} and homogeneity give
\begin{align}
 &\Tr\!\left[
 \sigma_{\A\B}
 \frac1\alpha
 \left(
 \mathrm D_{\frac1\alpha,\,
 a\sigma_{\A\B}^\alpha
 +b\left(
 r\I_{\A}\otimes\Tr_{\A}[\sigma_{\A\B}^\alpha]
 -\sigma_{\A\B}^\alpha
 \right)}
 \right)^{-1}
 \left[
 \left(
 r\I_{\A}\otimes\Tr_{\A}[\sigma_{\A\B}^\alpha]
 \right)^{\frac{1-\alpha}{\alpha}}
 \right]
 \right]\notag\\
 &\qquad\leq
 a^{-\frac{1-\alpha}{\alpha}}
 \Tr\!\left[
 \sigma_{\A\B}^\alpha
 \left(
 r\I_{\A}\otimes\Tr_{\A}[\sigma_{\A\B}^\alpha]
 \right)^{\frac{1-\alpha}{\alpha}}
 \right]\notag\\
 &\qquad=
 a^{-\frac{1-\alpha}{\alpha}}r^{\frac{1-\alpha}{\alpha}}Q_\alpha(\sigma).
 \label{eq:applied-anchor-calibration}
\end{align}
Multiplying by the same prefactor proves
\cref{eq:anchor-calibration}.
\end{proof}

\subsection{Bounds on the sandwiched supporting operator} \label{sec:sand_technical}

\begin{proof}[Proposition~\ref{prop:sand_calibrated}]
Let us denote
\begin{align*}
 Z_{\A\B}
 &:={\left[
 (\I_{\A}\otimes\sand\sigma_{\B}^{\frac{1-\alpha}{2\alpha}})
 \sigma_{\A\B}
 (\I_{\A}\otimes\sand\sigma_{\B}^{\frac{1-\alpha}{2\alpha}})
 \right]^{\alpha}},\\
 Z_{\B}&:=\Tr_{\A}[Z_{\A\B}],\\
 \sand Z_{\A\B}
 &:={\left[
 (\I_{\A}\otimes\sand\sigma_{\B}^{\frac{1-\alpha}{2\alpha}})
 \sand\tau
 (\I_{\A}\otimes\sand\sigma_{\B}^{\frac{1-\alpha}{2\alpha}})
 \right]^{\alpha}}.
\end{align*}
Then \cref{eq:comparison} reads
\begin{equation}
 \sand Z_{\A\B}=(a-b)Z_{\A\B}+rb\,\I_{\A}\otimes Z_{\B}.
 \label{eq:proof-comparison-root}
\end{equation}
Schmidt-rank domination gives
\begin{equation}
 Z_{\A\B}\leq r\,\I_{\A}\otimes Z_{\B}.
 \label{eq:proof-Schmidt}
\end{equation}
Consequently,
\begin{align}
 rb\,\I_{\A}\otimes Z_{\B}
 \leq\sand Z_{\A\B}
 \leq ra\,\I_{\A}\otimes Z_{\B},
 \qquad
 \sand Z_{\A\B}\geq aZ_{\A\B}.
 \label{eq:proof-root-orders}
\end{align}
Furthermore,
\begin{equation}
 \Tr_{\A}\sand [Z_{\A\B}]=[a+(D-1)b]Z_{\B},
 \qquad
 \Tr\sand Z_{\A\B}=[a+(D-1)b]\Tr [Z_{\A\B}].
 \label{eq:proof-root-traces}
\end{equation}
The normalized ratio in the fixed-point equation is therefore unchanged.
By the necessary-and-sufficient characterization in
\cref{lem:sand_supporting}, $\sand\sigma_{\B}$ is the unique optimizer at
$\sand\tau_{\A\B}$, and \cref{eq:comparison-value} follows from the second identity
in \cref{eq:proof-root-traces}.

Using \cref{eq:gradient,eq:proof-root-traces},
\begin{align}
\nabla\sand Q_{\alpha}(\sand\tau)
={}&[a+(D-1)b]^{\frac{1-\alpha}{\alpha}}
 (\I_{\A}\otimes Z_{\B}^{\frac{1-\alpha}{2\alpha}})
 \sand Z_{\A\B}^{-\frac{1-\alpha}{\alpha}}
 (\I_{\A}\otimes Z_{\B}^{\frac{1-\alpha}{2\alpha}}).
\label{eq:proof-score}
\end{align}
Because $0\leq\frac{1-\alpha}{\alpha}\leq1$, the map
$t\mapsto t^{-\frac{1-\alpha}{\alpha}}$ is operator monotone decreasing.  Applying it to
the first two inequalities in \cref{eq:proof-root-orders} proves
\cref{eq:score-window}.

For the anchor estimate, the last inequality in
\cref{eq:proof-root-orders} gives
\begin{align*}
\sand Z_{\A\B}^{-\frac{1-\alpha}{\alpha}}
\leq
a^{-\frac{1-\alpha}{\alpha}}
Z_{\A\B}^{-\frac{1-\alpha}{\alpha}}.
\end{align*}
The fixed-point identity gives
$\sand\sigma_{\B}=Z_{\B}/\Tr [Z_{\A\B}]$.  Since
\begin{align}
 (\I_{\A}\otimes\sand\sigma_{\B}^{\frac{1-\alpha}{2\alpha}})
 \sigma_{\A\B}
 (\I_{\A}\otimes\sand\sigma_{\B}^{\frac{1-\alpha}{2\alpha}})
 =Z_{\A\B}^{1/\alpha},
\end{align}
we can reconstruct the anchor as
\begin{align}
\sigma_{\A\B}
={}&(\Tr [Z_{\A\B}])^{\frac{1-\alpha}{\alpha}}
 (\I_{\A}\otimes Z_{\B}^{-\frac{1-\alpha}{2\alpha}})
 Z_{\A\B}^{1/\alpha}
 (\I_{\A}\otimes Z_{\B}^{-\frac{1-\alpha}{2\alpha}}).
\label{eq:proof-anchor-reconstruction}
\end{align}
Substituting this identity into the trace of \cref{eq:proof-score}, using
cyclicity, and observing that
$\frac{1}{\alpha}-\frac{1-\alpha}{\alpha}=1$, gives
\begin{align}
&\Tr\!\left[
(\I_{\A}\otimes Z_{\B}^{\frac{1-\alpha}{2\alpha}})
Z_{\A\B}^{-\frac{1-\alpha}{\alpha}}
(\I_{\A}\otimes Z_{\B}^{\frac{1-\alpha}{2\alpha}})
\sigma_{\A\B}
\right]=(\Tr [Z_{\A\B}])^{1/\alpha}
=\sand Q_{\alpha}(\sigma).
\end{align}
Combining this equality with the preceding anti-monotone estimate proves
\cref{eq:sand-anchor-calibration}.
\end{proof}

\bibliographystyle{amsalpha}
\bibliography{continuity, reference}

\end{document}